\documentclass[a4]{article}
\usepackage[latin1]{inputenc} 
\usepackage{a4wide}
\usepackage{times}
\usepackage{authblk} 
\usepackage{multicol}
\usepackage{verbatim}
\usepackage{amsmath,amsthm}
\usepackage{amssymb}
\usepackage{algpseudocode}
\usepackage{pgfpages}
\usepackage{graphicx}
\usepackage{ifthen}
\usepackage[titles]{tocloft}
\usepackage{mathrsfs}

\def\N {\ensuremath{\mathbb{N}}}

\def\C {\ensuremath{\mathbb{C}}}
\def\jac {\ensuremath{\mathsf{jac}}}
\def\X {\ensuremath{\mathbf{X}}}
\def\Y {\ensuremath{\mathbf{Y}}}
\def\x {\ensuremath{\mathbf{x}}}
\def\y {\ensuremath{\mathbf{y}}}
\def\a {\ensuremath{\mathbf{a}}}
\def\minor {\ensuremath{\mathsf{minor}}}
\def\Ker {\ensuremath{\mathsf{Ker}}}
\def\LM {\ensuremath{\mathsf{LM}}}
\def\LL {\ensuremath{\mathscr{L}}}
\def\NF {\ensuremath{\mathsf{NF}}}
\def\HS {\ensuremath{\mathrm{HS}}}
\def\ann {\ensuremath{\mathrm{ann}}}

\newtheorem{rem}{Remark}[section]
\newtheorem{lem}{Lemma}[section]
\newtheorem{thm}{Theorem}[section]

\newtheorem{prop}{Proposition}[section]

\newtheorem{cor}{Corollary}[section]
\newtheorem{defn}{Definition}[section]
\newtheorem{alg}{Algorithm}[section]
\newtheorem{conj}{Conjecture}[section]

\begin{document}

\title{Gr\"obner Bases of Bihomogeneous Ideals Generated by Polynomials of Bidegree $(1,1)$: Algorithms and Complexity}
\author{Jean-Charles Faug\`ere}
\author{Mohab Safey El Din}
\author{Pierre-Jean Spaenlehauer}
\date{}
\affil{INRIA, Paris-Rocquencourt, SALSA Project \\ 
UPMC, Univ Paris 06, LIP6 \\ 
CNRS, UMR 7606, LIP6 \\
UFR Ing\'enierie 919, LIP6 Passy-Kennedy \\
Case 169, 4, Place Jussieu, F-75252 Paris, France\\
\{Jean-Charles.Faugere, Mohab.Safey, Pierre-Jean.Spaenlehauer\}@lip6.fr}
\maketitle

\begin{abstract}
Solving multihomogeneous systems, as a wide range of \emph{structured
algebraic systems} occurring frequently in practical problems, is of
first importance. Experimentally, solving these systems with Gröbner
bases algorithms seems to be easier than solving homogeneous systems
of the same degree. Nevertheless, the reasons of this behaviour are
not clear. In this paper, we focus on bilinear systems
(i.e. bihomogeneous systems where all equations have bidegree
$(1,1)$). Our goal is to provide a theoretical explanation of the
aforementionned experimental behaviour and to propose new techniques
to speed up the Gröbner basis computations by using the
multihomogeneous structure of those systems. The contributions are
theoretical and practical. First, we adapt the classical $F_5$
criterion to avoid reductions to zero which occur when the input is a
set of bilinear polynomials. We also prove an explicit form of the
Hilbert series of bihomogeneous ideals generated by generic bilinear
polynomials and give a new upper bound on the degree of regularity of
generic affine bilinear systems. This leads to new complexity bounds
for solving bilinear systems. We propose also a variant of the $F_5$
Algorithm dedicated to multihomogeneous systems which exploits a
structural property of the Macaulay matrix which occurs on such
inputs. Experimental results show that this variant requires less time
and memory than the classical homogeneous $F_5$ Algorithm.

\end{abstract}

\section{Introduction}

The problem of multivariate polynomial system solving is an important
topic in computer algebra since algebraic systems can arise from many
practical applications (cryptology, robotics, real algebraic geometry,
coding theory, signal processing, etc...). One method to solve them is
based on the Gr\"obner bases theory.
Due to their practical importance, efficient algorithms to compute
Gr\"obner bases of algebraic systems are required: for instance
Buchberger's Algorithm \cite{buchberger}, Faug\`ere $F_4$
\cite{faugere139nea} or $F_5$ \cite{faugere2002nea}.

In this article, we focus on the $F_5$ Algorithm. In particular, the
$F_5$ criterion is a tool which removes the so-called \emph{reductions
  to zero} (which are useless) during the Gröbner basis computation
when the input system is a regular sequence. For instance, consider a
sequence of polynomials $(f_1,\ldots, f_m)$. The reductions to zero
come from the leading monomials in the colon ideals $\langle
f_1,\ldots,f_{i-1}\rangle : f_i$. Let $\LM(I)$ denote the ideal
generated by the leading monomials of the elements of an ideal
$I$. Then the reductions to zero detected by the $F_5$ criterion are
those related to $\LM(\langle f_1,\ldots,f_{i-1}\rangle)$. For
\emph{regular} systems, $\LM(\langle
f_1,\ldots,f_{i-1}\rangle)=\LM(\langle f_1,\ldots,f_{i-1}\rangle :
f_i)$. Therefore, the $F_5$ criterion removes all useless reductions. In
practice, if a homogeneous polynomial system is chosen ``at random'',
then it is regular. \smallskip

In this paper, we consider multihomo\-ge\-neous systems, which are not
regular. Such systems can appear in cryptography \cite{faugere2008},
in coding theory \cite{ourivski2002ntd} or in effective geometry
(see \cite{el2003polar, eldin2006}).

A multihomogeneous polynomial is defined with respect to a partition
of the unknowns, and is homogeneous with respect to
each subset of variables. The finite sequence of degrees is called the
\emph{multi-degree} of the polynomial. For instance, a bihomogeneous
polynomial $f$ of bi-degree $(d_1,d_2)$ over
$k[x_0,\ldots,x_{n_x},\penalty-1000 y_0,\ldots,y_{n_y}]$ is a polynomial such that
$$\forall \lambda,\mu, f(\lambda x_0,\ldots,\lambda x_{n_x},\mu
y_0,\ldots,\mu y_{n_y})=\lambda^{d_1} \mu^{d_2}
f(x_0,\ldots,x_{n_x},y_0,\ldots,y_{n_y}).$$ In general,
multihomogeneous systems are not regular. Consequently, the $F_5$
criterion does not remove all reductions to zero. Our goal is to
understand the underlying structure of these multihomogeneous
algebraic systems, and then use it to speed up the computation of a
Gröbner basis in the context of $F_5$.  In this paper, we focus on
bihomogeneous ideals generated by polynomials of bi-degree $(1,1)$.

\subsection{Main results}
Let $k$ be a field, $f_1,\ldots f_m \in
k[x_0,\ldots,x_{n_x},y_0,\ldots,y_{n_y}]$ be bilinear polynomials. We
denote by $F_i$ the polynomial family $(f_1,\ldots,f_i)$ and by $I_i$
the ideal $\langle F_i\rangle$. We start by describing the algorithmic
results of the paper, obtained by exploiting the algebraic structure
of bilinear systems.

In order to understand this structure, we study properties of the
jacobian matrices with respect to the two subsets of variables $x_0,
\ldots, x_{n_x}$ and $y_0, \ldots, y_{n_y}$:
$$
\jac_\x(F_i)=\left [\begin{array}{ccc}
  \frac{\partial f_1}{\partial x_0} & \cdots &  \frac{\partial f_1}{\partial x_{n_x}}\\
\vdots & \vdots & \vdots \\
  \frac{\partial f_i}{\partial x_0} & \cdots &  \frac{\partial f_i}{\partial x_{n_x}}\\
\end{array}\right ]
~\hspace{1cm}~
\jac_\y(F_i)=\left [\begin{array}{ccc}
  \frac{\partial f_1}{\partial y_0} & \cdots &  \frac{\partial f_1}{\partial y_{n_y}}\\
\vdots & \vdots & \vdots \\
  \frac{\partial f_i}{\partial y_0} & \cdots &  \frac{\partial f_i}{\partial y_{n_y}}\\
\end{array}\right ]
$$

We show that the kernels of those matrices (whose entries are linear
forms) correspond to the reductions to zero not detected by the
classical $F_5$ criterion. In general, all elements in these kernels
are vectors of maximal minors of the jacobian matrices (Lemma
\ref{lemmemi}). For instance, if $n_x=n_y=2$ and $m=4$,
consider $${\sf v}=(\minor(\jac_\x(F_4),1),-\minor(\jac_\x(F_4),2),
\minor(\jac_\x(F_4),3),-\minor(\jac_\x(F_4),4))$$ and $${\sf
  w}=(\minor(\jac_\y(F_4),1),-\minor(\jac_\y(F_4),2),
\minor(\jac_\y(F_4),3), -\minor(\jac_\y(F_4),4)),$$ where $\minor(\jac_\x(F_4),k)$ (resp. $\minor(\jac_\y(F_4),k)$) denotes
the determinant of the matrix obtained from $\jac_\x(F_4)$ (resp. $\jac_\y(F_4)$) by removing
the $k$-th column.  The generic \emph{syzygies} corresponding to
reductions to zero which are not detected by the classical $F_5$ criterion are
$${\sf v}\in \Ker_L(\jac_\x(F_4))
\text{ and } {\sf w}\in \Ker_L(\jac_\y(F_4)).$$ 

We show (Corollary \ref{corointro}) that, in general, the ideal
$I_{i-1}:f_i$ is spanned by $I_{i-1}$ and by the maximal minors of
$\jac_\x(F_{i-1})$ (if $i>n_y+1$) and $\jac_\y(F_{i-1})$ (if
$i>n_x+1$). The leading monomial ideal of $I_{i-1}:f_i$ describes the
reductions to zero associated to $f_i$. Thus we need results about
ideals generated by maximal minors of matrices whose entries are
linear forms in order to get a description of the syzygy module. In
particular, we prove that, in general, \emph{grevlex} Gröbner bases of
those ideals are linear combinations of the generators (Theorem
\ref{minorsGB}).  Based on this result, one can compute efficiently a
Gröbner basis of $I_{i-1} : f_i$ once a Gröbner basis of $I_{i-1}$ is known. 

\smallskip

This allows us to design an Algorithm (Algorithm \ref{algoprinc})
dedicated to bilinear systems, which yields an extension of the
classical $F_5$ criterion. This subroutine, when merged within a
matricial version of the $F_5$ Algorithm (Algorithm \ref{matrixF5}), eliminates
all reductions to zero during the computation of a Gröbner basis of a 
generic bilinear system. For instance, during the computation of a
grevlex Gröbner basis of a system of $12$ generic bilinear equations
over $k[x_0,\ldots,x_6,y_0,\ldots,y_6]$, the new criterion detects
$990$ reductions to zero which are not found by the usual $F_5$
criterion. Even if this new criterion seems to be more
complicated than the usual $F_5$ criterion (some
precomputations have to be performed), we prove that the overcost
induced by those precomputations is negligible compared to the cost of
the whole computation.

Next, we introduce a notion of \emph{bi-regularity} which describes
the structure of generic bilinear systems. When the input of Algorithm
\ref{algoprinc} is a bi-regular system, then it returns all reductions
to zero. We also give a complete description of the syzygy module of
such systems, up to a conjecture (Conjecture \ref{conjec}) on a
linear algebra problem over rings. This conjecture is supported by
practical experiments. We also prove that there are no reductions to
zero with the classical $F_5$ criterion for affine bilinear systems
(Proposition \ref{affreg}) which is important for practical
applications.

We describe now the main complexity results of the paper. We need some
results on the so-called Hilbert bi-series of ideals generated by
bilinear systems.  For bi-regular bilinear system, we give an explicit
form of this series (Theorem \ref{theohilb}):

$$\mathsf{HS}_{I_m}(t_1,t_2)=\frac{N_m}{(1-t_1)^{n_x+1} (1-t_2)^{n_y+1}},$$
$$\begin{array}{c}
N_m(t_1,t_2)~~~=~~~(1-t_1 t_2)^m +\\
\sum_{\ell=1}^{m-(n_y+1)}(1-t_1 t_2)^{m-(n_y+1)-\ell} t_1 t_2 (1-t_2)^{n_y+1}\big{[}1-(1-t_1)^\ell \sum_{k=1}^{n_y+1} t_1^{n_y+1-k}{{\ell+n_y-k}\choose {n_y+1-k}}\big{]} +\\
\sum_{\ell=1}^{m-(n_x+1)}(1-t_1 t_2)^{m-(n_x+1)-\ell} t_1 t_2 (1-t_1)^{n_x+1}\big{[}1-(1-t_2)^\ell \sum_{k=1}^{n_x+1} t_2^{n_x+1-k}{{\ell+n_x-k}\choose {n_x+1-k}}\big{]}.
\end{array}$$

\smallskip

After this analysis, we propose a variant of the Matrix $F_5$
Algorithm dedicated to multihomogeneous systems. The key idea is to
decompose the Macaulay matrices into a set of smaller matrices whose
row echelon forms can be computed independently. We provide some
experimental results of an implementation of this algorithm in
\texttt{Magma2.15}. This multihomogeneous variant can be more than
$20$ times faster for bihomogeneous systems than our \texttt{Magma}
implementation of the classical Matrix $F_5$ Algorithm. We perform a
theoretical complexity analysis based on the Hilbert series in the
case of bilinear systems, which provides an explanation of this gap.

Finally, we establish a sharp upper bound on the degree of regularity of
$0$-dimensional affine bilinear systems (Theorem \ref{thm:degregaff}). Let $f_1,\ldots, f_{n_x+n_y}$
be an affine bilinear system of $k[x_0,\ldots, x_{n_x-1},y_0,\ldots,
  y_{n_y-1}]$, then the maximal degree reached during the computation
of a Gröbner basis with respect to the grevlex ordering is upper
bounded by:
 $$\mathsf{d_{reg}}\leq \min\left(n_x+1,n_y+1\right).$$
 This bound is \emph{exact} in practice for generic bilinear systems
 and permits to derive complexity estimates for solving bilinear
 systems (Corollary \ref{coro:complaff}) which can be applied to
 practical problems (see for instance \cite{issacmr} for an
 application to the MinRank problem).

\subsection{State of the art}
The complexity analysis that we perform by proving properties on the
Hilbert bi-series of bilinear ideals follows a path which is similar
to the one used to analyze the complexity of the $F_5$ algorithm in
the case of homogeneous regular sequences (see \cite{bardet2005}). In
\cite{kreuzer2002basic}, the properties of Buchberger's Algorithm are
investigated in the context of multi-graded rings. 

The algorithmic use of multihomogeneous structures has been
investigated mostly in the framework of multivariate resultants (see
\cite{multires,EmirisM09} and references therein for the most recent
results) following the line of work initiated by
\cite{mccoy1933resultant}. In the context of solving polynomial
systems by using straight-line programs as data-structures,
\cite{jeronimo2007computing} provides an alternative way to compute
resultant formula for multihomogeneous systems.

As we have seen in the description of the main results, the knowledge
of Gr\"obner bases of ideals generated by maximal minors of linear
matrices play a crucial role. Theorem \ref{minorsGB} which states that
such Gr\"obner bases are obtained by a single row echelon form
computation is a variant of the main results in
\cite{sturmfels1993maximal} and \cite{bernstein93} (see also the
survey \cite{bruns2003grobner}).

More generally, the theory of {multihomogeneous elimination} is
investigated in \cite{remond} and \cite{remond1752geometrie} providing tools to
generalize some well-known notions (e.g. Chow forms, resultant
formula, heights) in the homogeneous case to multihomogeneous
situations. Such works are initiated in \cite{Waerden} where the
Hilbert bi-series of bihomogeneous ideals is introduced.

\subsection{Structure of the paper}
This paper is articulated as follows. Some tools from commutative
algebra are introduced. Next, we investigate the case of bilinear
systems and propose an algorithm to remove all reductions to zero
during the Gröbner basis computation. Then we prove its correctness
and explain why it is efficient for \emph{generic} bilinear
systems. To continue our study of the structure of bilinear ideals, we
give the explicit form of the Hilbert bi-series of generic bilinear
ideals. Finally, we prove a new bound on the degree of regularity of
generic affine bilinear systems and we use it to derive new complexity
bounds. Technical results and their proofs are postponed in Appendix.

\subsubsection*{Acknowledgments.}
We are grateful to Ludovic Perret and Ioannis Z. Emiris for their helpful comments and suggestions.

\newpage
\tableofcontents
\section{Gr\"obner bases: the Matrix $F_5$ Algorithm}

\subsection{Gr\"obner bases: notations}

In this section, $R$ denotes the ring $k[x_1,\ldots, x_n]$
(where $k$ is a field) and for all $\beta=(\beta_1, \ldots, \beta_n)
\in\N^n$, $x^\beta$ denotes $x_1^{\beta_1},\cdots,
x_n^{\beta_n}$. Gr\"obner bases are defined with respect to a
monomial ordering (see \cite{cox}, page 55, Definition 1).  In this
paper, we focus in particular on the so-called {\em grevlex} ordering
(degree reverse lexicographical ordering).

\begin{defn}\label{def:grevlex} The \emph{grevlex} ordering is defined by: $$x^\alpha \prec
  x^\beta \Leftrightarrow \begin{cases} \sum \alpha_i < \sum
    \beta_i\text{ or }\\ \sum \alpha_i = \sum \beta_i \text{~and the
      first coordinates }\\\text{ from the right which are different
      satisfy~} \alpha_i>\beta_i. \end{cases}$$
\end{defn}
If $\prec$ is a monomial ordering  and
$f\in R$ is a polynomial, then its greatest monomial with respect to
$\prec$ is called \emph{leading monomial} and denoted by $\LM_\prec(f)$
(or simply $\LM(f)$ when there is no ambiguity on the considered
ordering).

If $I\subset R$ is a polynomial ideal, its \emph{leading monomial
  ideal} (i.e. $\left\langle \{ \LM_\prec(f) : {f\in I}\}\right\rangle$) is denoted
by $\LM_\prec(I)$ (or simply $\LM(I)$ when there is no ambiguity
on the ordering) .

\begin{defn} let $I\subset R$ be an ideal, and $\prec$ be a
  monomial ordering. A \emph{Gr\"obner basis} of $I$
  (relatively to $\prec$) is a finite subset $G\subset I$ such that:
  $\langle \LM_\prec(G) \rangle = \LM_\prec(I)$. 
\end{defn}

\begin{defn}
  Let $I\subset R$ be an ideal, $\prec$ be a monomial
  ordering and $f\in R$ be a polynomial.  Then there exist unique
  polynomials $\tilde f\in R$ and $g\in I$ such that $f=\tilde f +g$, $\tilde f$ is monic
  and none of the monomials appearing in $\tilde{f}$ are in $\LM_\prec(I)$.
  The polynomial $\tilde f$ is called the \emph{normal form} of $f$ (with
  respect to $I$ and $\prec$), and is denoted $\NF_{I,\prec}(f)$.
\end{defn}
It is well known that $\NF_{I, \prec}(f)=0$  if and only if $f\in I$ (see e.g. \cite{cox}).
\begin{defn} Let $I\subset R$ be an homogeneous ideal, $\prec$ be a monomial ordering and $D$ be an integer. We call
  $D$-Gröbner basis a finite set of polynomials $G$ such that $\langle
  G\rangle=I$ and
$$\forall f\in I \text{ with } \deg(f)\leq D,\text{ there exists } g\in G \text{ such that } \LM_\prec(g) \text{ divides } \LM_\prec(f).$$
\end{defn}

The following Lemma is a straightforward consequence of Dickson's Lemma \cite[page 71, Theorem 5]{cox}.
\begin{lem} Let $I\subset R$ be an ideal and let $\prec$ be a monomial ordering. There exists $D\in\mathbb{N}$ such
  that every $D$-Gröbner basis with respect to $\prec$ is a Gröbner
  basis of $I$ with respect to $\prec$.
\end{lem}

\subsection{The Matrix $F_5$ Algorithm}
We use a variant of the $F_5$ Algorithm,
called Matrix $F_5$ Algorithm, which is suitable to perform
complexity analyses (see \cite{bardet,bardet2005,faurah}).

Given a set of generators $(f_1, \ldots, f_m)$ of an homogeneous
polynomial ideal $I\subset R$, an integer $D$ and a
monomial ordering $\prec$, the Matrix $F_5$ Algorithm computes a
$D$-Gröbner basis of $I$ with respect to $\prec$.
It performs incrementally by considering the ideals $I_i=\langle f_1,
\ldots, f_i\rangle$ for $1\leq i \leq m$.  

Let $d\in \N$, denote by
$R_d$ the $k$-vector space of polynomials in $R$ of degree $d$. As in \cite{faugere2002nea} and \cite{bardet}, we use a
definition of the row echelon form of a matrix which is slightly
different from the usual definition: we call \emph{row echelon form}
the matrix obtained by applying the Gaussian elimination Algorithm
\emph{without permuting the rows}.
The idea of the Matrix $F_5$ Algorithm (see Algorithm \ref{matrixF5}
below) is to calculate triangular bases of the vector spaces $I_i\cap
R_d$ for $1\leq d\leq D$ and $1\leq i\leq m$ and to deduce from them a
$d$-basis of $I_{i+1}$. These triangular bases are obtained by
computing row echelon forms of the Macaulay matrices.

In the algorithm which follows, the columns in the matrix $\mathcal
M_{d,i}$ correspond to the monomials of $R$ of degree $d$ and are
sorted by the chosen monomial ordering $\prec$ (from the largest to
the smallest). An homogeneous polynomial is identified with the
corresponding row in the matrix. Each row has a signature $(t,f_j)$,
where $t$ is a monomial and $1\leq j \leq i$. The rows of the matrices
are sorted as follows: a row with signature $(t_1,f_j)$ is preceding a
row with signature $(t_2,f_k)$ if $j<k$ or ($j=k$ and $t_1\prec t_2$).

When the row echelon form of a matrix is computed, 
the rows which are linear combinations of
preceding rows are reduced to zero. Such computations are useless:
removing these rows before computing the row echelon form will not
modify the result but lead to significant practical
  improvements. 
The so-called $F_5$ criterion (see \cite{faugere2002nea}) is used to detect
these \emph{reductions to zero} and is given below.

\begin{alg}
{\bf $F_5$\textsf{criterion}} - returns a boolean
\begin{algorithmic}[1]

\Require $\begin{cases}
(t,f_i) \text{ the signature of a row}\\
\text{A matrix }\mathcal{M}\text{ in row echelon form} 
\end{cases}$
\State Return ($t$ is the leading monomial of a row of $\mathcal{M}$)
\end{algorithmic}
\end{alg}
 
Now, one gives a description of the Matrix $F_5$ Algorithm.

\begin{alg}\label{matrixF5}
{\bf \textsf{Matrix} $F_5$} (see \cite{bardet,faugere2002nea})
\begin{algorithmic}[1]
  \Require $\begin{cases}(f_1,\ldots,f_m)\text{ homogeneous polynomials of degree } d_1\leq d_2\leq\ldots\leq d_m\\
    D\text{ an integer}\\
    \text{a monomial ordering }\prec\end{cases}$ \Ensure $G$ is a $D$-Gröbner basis of $\langle f_1,\ldots,f_m\rangle$
  for $\prec$ \State $G\leftarrow\emptyset$ \For {$d$ from $d_1$ to
    $D$} 
\State $\widetilde{{\cal M}_{d,0}}\leftarrow$  matrix with $0$ rows
\For{$i$ from $1$ to $m$} \State Construct ${\cal
    M}_{d,i}$ by adding to $\widetilde{{\cal M}_{d,i-1}}$ the following rows:
\If{$d_i=d$} 
\State add the row $f_i$ with signature $(1,f_i)$
\EndIf
\If{$d>d_i$}
\State for all $f$ from $\widetilde{{\cal
      M}_{d-1,i}}$ with signature $(e,f_i)$, such that $x_\lambda$ 
is the \State greatest variable of $e$, add the $n-\lambda+1$ rows
  $x_\lambda f,x_{\lambda+1} f,\ldots,x_n f$ with the \State signatures $(x_\lambda e, f_i),(x_{\lambda+1}e, f_i),\ldots,(x_n e,
  f_i)$ 
except those which satisfy: 
\State $F_5${\bf\textsf{criterion}} $( (x_{\lambda+k} e,f_i) , \widetilde{{\cal M}_{d-d_i,i-1}})$=true
\EndIf
\State Compute $\widetilde{{\cal M}_{d,i}}$ the row echelon form of
${\cal M}_{d,i}$ \State Add to $G$ the polynomials corresponding to
rows of $\widetilde{{\cal M}_{d,i}}$ such that their \State leading monomial
is different from the leading monomial of \State the row with same signature in
${\cal M}_{d,i}$
  \EndFor
\EndFor
\State return $G$
\end{algorithmic}
\end{alg}

We recall now some results mostly given by \cite{faugere2002nea} which justify
the $F_5$ criterion by relating reductions to zero appearing in an
incremental computation of a Gr\"obner basis of a homogeneous ideal
with the syzygy module of the polynomial system under consideration.

\begin{defn} Let $(f_1,\ldots, f_m)$ be polynomials of $R$. A syzygy
  is an element $s=(s_1,\ldots, s_m)\in R^m$ such that $\sum_{j=1}^m
  f_j s_j=0$.  The degree of the syzygy is defined by
  $\max_j(\deg(f_j)+\deg(s_j))$.  The set of all syzygies is a
  submodule of $R^m$ called the \emph{syzygy module} of
  $(f_1,\ldots,f_m)$.
\end{defn}

The next theorem explains how reductions to zero and syzygies are related:
\begin{thm}[$F_5$ criterion, \cite{faugere2002nea}]~\label{thm:F5}
\begin{enumerate}
\item If $t\in \LM(I_{i-1})$ then there exists a syzygy
  $(s_1,\ldots,s_i)$ of $(f_1, \ldots, f_i)$ such that $\LM(s_i)=t$.
\item Let $(t,f_i)$ be the signature of a row
  of ${\cal M}_{d,m}$. Then the following assertions are equivalent:
\begin{enumerate}
\item the row $(t,f_i)$ is zero in the row echelon form $\widetilde{{\cal M}_{d,m}}$.
\item $t\notin LM(I_{i-1})$ and there exists a syzygy $s=(s_1,\ldots,s_i)$ of $(f_1,\ldots,f_i)$ such that $t=\LM(s_i)$.
\end{enumerate}
\end{enumerate}
\end{thm}

The rows eliminated by the $F_5$ criterion correspond to the trivial
syzygies, i.e. the syzygies $(s_1,\ldots, s_m)$ such that $\forall 1\leq i\leq m$, $s_i\in \langle f_1,\ldots, f_{i-1},f_{i+1},\ldots,f_m\rangle$. These particular syzygies come from the commutativity of $R$
(for all $ 1\leq i,j\leq m$, $ f_i f_j - f_j f_i=0$). It is well known
that in the generic case, 
the syzygy module of a polynomial system is
generated by the trivial syzygies.

\begin{defn}\cite[page 419]{eisenbud}\label{def:regular} Let $(f_1,\ldots,f_m)$ be a sequence of
  homogeneous polynomials and let $I_i\subset R$ be the ideal $\langle
  f_1,\ldots,f_i\rangle$. The following assertions are equivalent:
\begin{enumerate}
\item the syzygy module of $(f_1,\ldots,f_m)$ is generated by the trivial syzygies.
\item for $2\leq i\leq m$, $f_i$ is not a divisor of $0$
  in $R/I_{i-1}$.
\end{enumerate}
A sequence of polynomials which satisfies these conditions is called a
\emph{regular sequence}.
\end{defn}

This notion of regularity is essential since the regular sequences
correspond exactly to the systems such that there is no reduction to
zero during the computation of a Gröbner basis with $F_5$ (see
\cite{faugere2002nea}). Moreover, generic polynomial systems are regular.

\section{Gröbner bases computation for bilinear systems}

\subsection{Overview}
Let $F=(f_1, \ldots, f_4)$ be a
sequence of four bilinear polynomials in
$\mathbb{Q}[x_0,x_1,x_2,y_0,y_1,y_2]$, $I$ be the ideal generated by $F$
and $V\subset \C^6$ be its associated algebraic variety. 
As above, $I_i$ denotes the ideal $\langle f_1,\ldots, f_i\rangle$, and we consider the grevlex ordering with $x_0\succ\ldots\succ x_{n_x}\succ y_0\succ\ldots\succ y_{n_y}$.
Since $f_1, \ldots, f_4$ are bilinear, for all $(a_0, a_1, a_2)\in
\C^3$ and $1\leq i \leq 4$, $f_i(a_0, a_1, a_2, 0, 0, 0)=0$. Hence,
$V$ contains the linear affine subspace defined by $y_0=y_1=y_2=0$
which has dimension $3$. We conclude that $V$ has dimension at least
$3$.

Consequently, the sequence $(f_1,f_2,f_3,f_4)$ is not regular
(since the co-dimension of an ideal generated by a regular sequence is
equal to the length of the sequence). Hence, there are reductions to zero during
the computation of a Gröbner basis with the $F_5$ Algorithm
(see \cite{faugere2002nea}). 

When the four polynomials are chosen randomly, one remarks
experimentally that these reductions correspond to the rows with
signatures $(x_0^3, f_4)$ and $(y_0^3, f_4)$. 
This experimental
observation can be explained as follows.

Consider the jacobian matrices
$$\jac_\x(F)=
\left [\begin{array}{ccc}
  \frac{\partial f_1}{\partial x_0} &   \frac{\partial f_1}{\partial x_1} &   \frac{\partial f_1}{\partial x_2} \\
\vdots & \vdots & \vdots \\
  \frac{\partial f_4}{\partial x_0} &   \frac{\partial f_4}{\partial x_1} &   \frac{\partial f_4}{\partial x_2} \\
\end{array}\right ]
\hspace{0.5cm}\text{ and }\hspace{0.5cm}
\jac_\y(F)=
\left [\begin{array}{ccc}
  \frac{\partial f_1}{\partial y_0} &   \frac{\partial f_1}{\partial y_1} &   \frac{\partial f_1}{\partial y_2} \\
\vdots & \vdots & \vdots \\
  \frac{\partial f_4}{\partial y_0} &   \frac{\partial f_4}{\partial y_1} &   \frac{\partial f_4}{\partial y_2} \\
\end{array}\right ]
$$
and the vectors of variables $\X$ and $\Y$.  By Euler's formula, it is
immediate that for any sequence of polynomials $(q_1, q_2, q_3, q_4)$,
\begin{equation}
  \label{eq:1}
  (q_1, \ldots, q_4).\jac_\x(F).\X=\sum_{i=1}^4q_if_i\hspace{0.2cm}\text{ and }\hspace{0.2cm}
(q_1, \ldots, q_4).\jac_\y(F).\Y=\sum_{i=1}^4q_if_i
\end{equation}
Denote by $\Ker_L(\jac_\x(F))$ (resp. $\Ker_L(\jac_\y(F))$) the left kernel of $\jac_\x(F)$ (resp. $\jac_\y(F)$).

Therefore, if $(q_1, \ldots, q_4)$ belongs to
$\Ker_L(\jac_\x(F))$ (resp. $\Ker_L(\jac_\y(F))$), then the relation
(\ref{eq:1}) implies that $(q_1, \ldots, q_4)$ belongs to the
syzygy module of $I$.

Given a $(k+1,k)$-matrix ${\sf M}$, denote by ${\sf minor}({\sf M},
j)$ the minor obtained by removing the $j$-th row from ${\sf M}$. Consider 
$${\sf v}=(\minor(\jac_\x(F),1),-\minor(\jac_\x(F),2),
\minor(\jac_\x(F),3),-\minor(\jac_\x(F),4)).$$
By
Cramer's rule, it is straightforward to prove that 
$\mathsf v\in \Ker_L(\jac_\x(F)).$
A symmetric statement can be made for $\jac_\y(F)$. {F}rom this observation, one
deduces that $\minor(\jac_\x(F),4)f_4$ (resp.
$\minor(\jac_\y(F),4)f_4$) belongs to $I_3=\langle f_1, f_2,
f_3\rangle$. 

We conclude that the rows with signature
$$(\LM(\minor(\jac_\x(F),4)),f_4) \text{~~~ and ~~~}
(\LM(\minor(\jac_\y(F),4)),f_4)$$ are reduced to zero when performing
the Matrix $F_5$ Algorithm described in the previous section.  A
straightforward computation shows that if $F$ contains polynomials
which are chosen randomly, then $$\LM(\minor(\jac_\x(F),4))=y_0^3\text{~~~ and ~~~} \LM(\minor(\jac_\y(F),4))=x_0^3.$$

In this section, we generalize this approach to sequences of bilinear
polynomials of arbitrary length. Hence, the jacobian matrices have a
number of rows which is is not the number of columns incremented by
$1$. But, even in this more general setting, we exhibit a a
relationship between the left kernels of the jacobian
matrices and the syzygy module of the ideal spanned by the sequence
under consideration. This allows us to prove a new $F_5$-criterion
dedicated to bilinear systems. On the one hand, when plugged into the
Matrix $F_5$ Algorithm, this criterion detects reductions to zero
which are not detected by the classical criterion. On the other hand,
we prove that a $D$-Gr\"obner basis is still computed by the Matrix
$F_5$ Algorithm when it uses the new criterion.

\subsection{Jacobian matrices of bilinear systems and syzygies}
\label{sect:algo}
From now on, we use the following notations:
\begin{itemize}
\item $R=k[x_0, \ldots, x_{n_x}, y_0, \ldots, y_{n_y}]$;
\item $F=(f_1, \ldots, f_m)\subset R^m$ is a sequence of bilinear
  polynomials and $F_i=(f_1, \ldots, f_i)$ for $1\leq i \leq m$;
\item $I$ is the ideal generated by $F$ and $I_i$ is the ideal generated by $F_i$;
\item Let $\mathsf{M}$ be a $\ell\times c$ matrix, with $\ell>c$. We call \emph{maximal minors} of $\mathsf{M}$ the determinants of the $c\times c$ sub-matrices of $\mathsf{M}$;
\item $\jac_\x(F_i)$ and $\jac_\y(F_i)$ are respectively the jacobian matrices 
$$
\left [\begin{array}{ccc}
  \frac{\partial f_1}{\partial x_0} & \cdots &  \frac{\partial f_1}{\partial x_{n_x}}\\
\vdots & \vdots & \vdots \\
  \frac{\partial f_i}{\partial x_0} & \cdots &  \frac{\partial f_i}{\partial x_{n_x}}\\
\end{array}\right ]
\text{ and }
\left [\begin{array}{ccc}
  \frac{\partial f_1}{\partial y_0} & \cdots &  \frac{\partial f_1}{\partial y_{n_y}}\\
\vdots & \vdots & \vdots \\
  \frac{\partial f_i}{\partial y_0} & \cdots &  \frac{\partial f_i}{\partial y_{n_y}}\\
\end{array}\right ];
$$
\item Given a matrix ${\sf M}$, $\Ker_L({\sf M})$ denotes the left
  kernel of ${\sf M}$;
\item $\X$ is the vector of variables $[x_0, \ldots, x_{n_x}]^t$ and
  $\Y$ is the vector of variables $[y_0, \ldots, y_{n_y}]^t$;
\item $(f_1,\ldots,f_m)\in k[x_0,\ldots,x_{n_x-1},y_0,\ldots,y_{n_y-1}]^m$ is called \emph{affine bilinear system} if there exists an homogeneous bilinear system $(f_1^h,\ldots, f_m^h)\in k[x_0,\ldots,x_{n_x},y_0,\ldots,y_{n_y}]^m$ such that
$$f_i(x_0,\ldots,x_{n_x-1},y_0,\ldots,y_{n_y-1})=f_i^h(x_0,\ldots,x_{n_x-1},1,y_0,\ldots,y_{n_y-1},1).$$
\end{itemize}

\begin{lem}\label{lemmemi}
  Let $i>n_x+1$ (resp. $i>n_y+1$), and let ${\frak s}$ be a maximal minor of
  $\jac_\x(F_{i-1})$ (resp. $\jac_\y(F_{i-1})$). Then there exists a vector
  $(s_1, \ldots, s_{i-1},
  {\frak s})$ in $\Ker_L(\jac_\x(F_i))$
  (resp. $\Ker_L(\jac_\y(F_i))$).
\end{lem}
\noindent{\em Proof.}
  The proof is done when considering ${\frak s}$ as a maximal minor of
  $\jac_\x(F_{i-1})$ with $i>n_x+1$. The case where ${\frak s}$ is a
  maximal minor of $\jac_\y(F_{i-1})$ with $i>n_y+1$ is proved similarly.

  Note that $\jac_\x(F_{i-1})$ is a matrix with $i-1$ rows and $n_x+1$
  columns and $i-1\geq n_x+1$. Denote by $(j_1, \ldots, j_{i-n_x-2})$ the
  rows deleted from $\jac_\x(F_{i-1})$ to construct its submatrix $J$
  whose determinant is ${\frak s}$.

  Consider now the $i\times(i-n_x-2)$-matrix ${\sf T}$ such that its
  $(\ell, k)$ entry is $1$ if and only if $\ell=j_k$ else it is $0$. $N$ denotes the following $i\times (i-1)$ matrix:
$${\sf N}=\left
  [\begin{array}{c|c}\jac_\x(F_i) & {\sf T}\end{array}\right ].$$
A straightforward use of Cramer's rule shows that 
$$(\minor({\sf N}, 1), -\minor({\sf N}, 2), \ldots, (-1)^{i+1}\minor({\sf N}, i))\in \Ker_L(\mathsf N).$$
Remark that this implies 
$$(\minor({\sf N}, 1), -\minor({\sf N}, 2), \ldots, (-1)^{i+1}\minor({\sf N}, i))\in \Ker_L(\jac_\x(F_i)).$$
A routine computation of $\minor({\sf N}, i)$ by going across the last columns of $N$ shows that $\minor({\sf N},i)=\pm {\frak s}$

\hfill $\square$
\begin{thm} \label{correct} 
  Let $i>n_x+1$ (resp. $i>n_y+1$) and let $s$
  be a linear combination of maximal minors of $\jac_\x(F_{i-1})$
  (resp. $\jac_\y(F_{i-1})$). Then $s\in I_{i-1} : f_{i}$.
\end{thm}

\noindent{\em Proof.}
  By assumption, $s=\sum_\ell a_\ell\, {\frak s}_\ell$ where each ${\frak
    s}_\ell$ is a maximal minor of
  $\jac_\x (F_{i-1})$. 
  According to Lemma
  \ref{lemmemi}, for each minor ${\frak s}_\ell$ there exists
  $(s^{(\ell)}_1,\ldots, s^{(\ell)}_{i-1})$ such that
$$(s^{(\ell)}_1,\ldots,s^{(\ell)}_{i-1}, {\frak s}_\ell)\in \Ker_L(\jac_\x(F_i))$$
Thus, by summation over $\ell$, one obtains
\begin{equation}
  \label{eq:2}
  (\sum_\ell a_\ell s^{(\ell)}_1,\ldots,\sum_\ell a_\ell s^{(\ell)}_{i-1}, s)\in \Ker_L(\jac_\x(F_i)).
\end{equation}
Moreover, by Euler's formula
$$  (\sum_\ell a_\ell s^{(\ell)}_1,\ldots,\sum_\ell a_\ell s^{(\ell)}_{i-1}, s)\jac_\x(F_i)\X= s\, f_i+\sum_{j=1}^{i-1}\left (\sum_\ell a_\ell s^{(\ell)}_j\right ) f_j.$$
By the relation (\ref{eq:2}), $s\, f_i+\sum_{j=1}^{i-1}\left
  (\sum_\ell a_\ell s^{(\ell)}_j\right ) f_j=0$, which implies that
$s\in I_{i-1}:f_{i}$.
\hfill $\square$

\begin{cor}
  Let $i>n_x+1$ (resp. $i>n_y+1$), $M^{(i)}_\x$ (resp. $M^{(i)}_\y$)
  be the ideal generated by the maximal minors of $\jac_\x(F_i)$
  (resp. $\jac_\y(F_i)$). Then $M_\x^{(i-1)}\subset I_{i-1}: f_i$
  (resp. $M_\y^{(i-1)}\subset I_{i-1}: f_i$).
\end{cor}
\noindent{\em Proof.}
  By Theorem \ref{correct}, all minors of $\jac_\x(F_{i-1})$
  (resp. $\jac_\y(F_{i-1})$) are elements of $I_{i-1}:f_i$. Thus,
  $I_{i-1}:f_i$ contains a set of generators of $M_\x^{(i-1)}$
  (resp. $M_\y^{(i-1)}$). Since $I_{i-1}:f_i$ is an ideal, our assertion
  follows.
\hfill $\square$

The above result implies that for all $g\in M_\x^{(i-1)}$ (resp. $g\in
M_\y^{(i-1)}$), the rows of signature $({\sf \LM}(g), f_i)$ are
reduced to zero during the Matrix $F_5$ Algorithm. In order to remove
these rows, it is crucial to compute a Gr\"obner basis of the ideals
$M_\x^{(i-1)}$ and $M_\y^{(i-1)}$. These ideals are generated by the
maximal minors of matrices whose entries are linear forms. The goal
of the following section is to understand the structure of such ideals
and how Gr\"obner bases can be efficiently computed in that case.

\subsection{Gr\"obner bases and maximal minors of matrices with linear entries}

Let $\mathscr{L}$ be the set of homogeneous linear forms in the
ring $R_\X=k[x_0, \ldots, x_{n_x}]$, $\prec$~be the
$grevlex$ ordering on $R_\X$ (with $x_0\succ \cdots \succ x_{n_x}$)
and ${\sf Mat}_\LL(p,q)$ be the set of $p\times q$ matrices with entries
in $\mathscr{L}$ with $p\geq q$ and $n_x\geq p-q$. Note that ${\sf Mat}_\LL(p,q)$ is a
$k$-vector space of finite dimension.

Given ${\sf M}\in {\sf Mat}_\LL(p,q)$, we denote by ${\sf
  MaxMinors}({\sf M})$ the set of maximal minors of ${\sf M}$.  We
denote by ${\sf Macaulay}_\prec({\sf MaxMinors}({\sf M}),q)$ the
Macaulay matrix in degree $q$ associated to ${\sf MaxMinors}({\sf M})$ and to the
ordering $\prec$ (each row represents a polynomial of ${\sf
  MaxMinors}({\frak M})$ and the columns represent the monomials of
degree $q$ of $k[x_0,\ldots,x_{n_x}]$ sorted by $\prec$ from the
largest to the smallest).

The main result of this paragraph lies in the following theorem: it
states that, in general, a Gr\"obner basis of $\langle {\sf
  MaxMinors}({\sf M})\rangle$ is a \emph{linear} combination of the
generators.

\begin{thm}\label{minorsGB}
  There exists a nonempty Zariski-open set $O$ in ${\sf
    Mat}_\LL(p,q)$ such that for all ${\sf M}\in O$, a $grevlex$
  Gr\"obner basis of $\langle {\sf MaxMinors}({\sf M})\rangle$ with
  respect to $\prec$ is obtained by computing the row echelon form of
  ${\sf Macaulay}_\prec({\sf MaxMinors}({\sf M}),q)$.
\end{thm}

This theorem is related with a result from Sturmfels, Bernstein and Zelevinsky (1993), which
states that the ideal generated by the maximal minors of a matrix whose
entries are variables is a universal Gröbner Basis. We tried without success to use this result in order to prove Theorem \ref{minorsGB}. Therefore, we propose an ad-hoc proof, 
which is based on the following
Lemmas whom proofs are postponed at the end of the paragraph. 

\begin{lem}\label{lemGB1}
  Let ${\sf Monomials}_{p-q}(q)$ be the set of monomials of degree $q$ in
  $k[x_0, \ldots, x_{p-q}]$. There exists a Zariski-open subset $O'$
  of ${\sf Mat}_\LL(p,q)$ such that for all ${\sf M}\in O'$
$$\langle {\sf Monomials}_{p-q}(q)\rangle\subset\LM(\langle{\sf MaxMinors}({\sf M}) \rangle ) $$
\end{lem}

\begin{lem}\label{lemGB2}
  Let ${\sf Monomials}_{p-q}(q)$ be the set of monomials of degree $q$ in
  $k[x_0, \ldots, x_{p-q}]$. There exists a Zariski-open subset $O''$
  of ${\sf Mat}_\LL(p,q)$ such that for all ${\sf M}\in O''$
$$\LM(\langle{\sf MaxMinors}({\sf M}) \rangle )\subset\langle {\sf Monomials}_{p-q}(q)\rangle $$
\end{lem}

\begin{lem}\label{lemGB3}
  The Zariski-open set $O'\cap O''\subset {\sf Mat}_\LL(p,q)$ is nonempty. 
\end{lem}

\noindent\emph{Proof of Theorem \ref{minorsGB}.}
From Lemmas \ref{lemGB1}, \ref{lemGB2} and \ref{lemGB3}, $O=O'\cap O''$ is a nonempty Zariski open set. Now let $\mathsf{M}$ be a matrix in $O\subset {\sf Mat}_\LL(p,q)$.
$$\langle {\sf Monomials}_{p-q}(q)\rangle=\LM(\langle{\sf MaxMinors}({\sf M}) \rangle ).$$
Thus all polynomials in a minimal Gröbner basis of $\langle{\sf MaxMinors}({\sf M}) \rangle$ have degree $q$ and then can be obtained by computing the row echelon form of $\mathsf{Macaulay}_\prec(\mathsf{MaxMinors(M)},q)$.
\hfill$\square$

We prove now Lemmas \ref{lemGB1}, \ref{lemGB2} and \ref{lemGB3}. 

\noindent{\em Proof of Lemma \ref{lemGB1}.}  Let ${\frak M}$ be the $(p,
q)$-matrix whose $(i,j)$-entry is a generic homogeneous linear form
$\sum_{k=0}^{n_x}{\frak a}_k^{(i,j)}x_k\in k(\mathfrak a_0^{(i,j)}, \ldots,
\mathfrak a_k^{(i,j)})[x_0, \ldots, x_{n_x}]$. Denote by $${\frak a}=\{{\frak
  a}_k^{(i,j)}, 0\leq k\leq n_x,\; 1\leq i\leq p,\; 1\leq j\leq q\}$$
and given a set 
$${\bf a}=\{{\bf
  a}_k^{(i,j)}\in k, 0\leq k\leq n_x,\; 1\leq i\leq p,\; 1\leq j\leq
q\}$$ consider the specialization map $\varphi_\a: {\frak
  M}\mapsto {\frak M}_\a\in {\sf Mat}_\LL(p,q)$ such that the
$(i,j)$-entry of ${\frak M}_\a$ is
$\sum_{k=0}^{n_x}{\a}_k^{(i,j)}x_k\in k[x_0, \ldots, x_{n_x}]$. We
prove below that there exists a polynomial $g\in k[{\frak a}]$ such
that, if $g(\a)\neq 0$ then $$\langle {\sf
  Monomials}_{p-q}(q)\rangle\subset \LM(\langle {\sf
  MaxMinors}(\varphi_\a({\frak M}))\rangle ).$$

 Consider the Macaulay matrix ${\sf Macaulay}_\prec({\sf
    MaxMinors}({\frak M}),q)$.

  Remark that the number of monomials in ${\sf Monomials}_{p-q}(q)$
  equals the number of maximal minors of ${\frak M}$. Moreover, by
  construction of ${\sf Macaulay}_\prec({\sf MaxMinors}({\frak M}),q)$
  and by definition of $\prec$ (see Definition \ref{def:grevlex}), the first
  ${{p}\choose{q}}$ columns of ${\sf Macaulay}_\prec({\sf
    MaxMinors}({\frak M}),q)$ contain the coefficients of the monomials
  in ${\sf Monomials}_{p-q}(q)$ of the polynomials in ${\sf
    MaxMinors}({\frak M})$.

  Saying that $\langle {\sf Monomials}_{p-q}(q)\rangle\subset
  \LM(\langle {\sf MaxMinors}({\frak M}) \rangle)$ is equivalent to
  saying that the determinant of the square submatrix of ${\sf
    Macaulay}_\prec({\sf MaxMinors}({\frak M}),q)$ containing its first
  ${p}\choose{q}$ columns is non-zero. Let $g\in k[{\frak a}]$ be this
  determinant. 

  The inequation $g\neq 0$ defines a Zariski-open
  set ${O'}$ such that for all $\a\in {O'}$
$$\langle {\sf  Monomials}_{p-q}(q)\rangle\subset \LM(\langle {\sf
  MaxMinors}(\varphi_\a({\frak M}))\rangle ). $$
\hfill $\square$

In the following $\psi$ denotes the canonical inclusion morphism from $k[x_0,\ldots, x_{n_x}]$ to $k'[x_0,\ldots, x_{p-q}]$, where $k'$ is the field of fractions $k(x_{p-q+1},\ldots, x_{n_x})$.

  For $(v_1,\ldots, v_{n_x-p+q})$, $\psi_{\mathbf{v}}$ denotes the specialization morphism:
$$\begin{array}{cccc}
\psi_{\mathbf{v}}:&k[x_0,\ldots,x_{n_x}]&\longrightarrow&k[x_0,\ldots,x_{p-q}]\\
&f(x_0,\ldots,x_{n_x})&\longmapsto&f(x_0,\ldots,x_{p-q},v_1,\ldots,v_{n_x-p+q})
\end{array}$$

\begin{lem}\label{lem:degproj}
There exists a Zariski open set $O'''$, such that if $\mathbf a\in O'''$, then the ideal $\langle \mathsf{MaxMinors}(\psi\circ\varphi_{\mathbf a}(\mathfrak M))\rangle$ is radical and its degree is $p\choose {q-1}$.
\end{lem}

\noindent{\em Proof.}
There exists an affine bilinear system $f_1,\ldots, f_p\in k'(\mathfrak a)[x_0,\ldots,x_{p-q},y_0,\ldots,y_{q-2}]$, such that:
$$
\psi(\mathfrak{M})\cdot\begin{pmatrix}
y_0\\\vdots\\y_{q-2}\\1
\end{pmatrix}=\begin{pmatrix}f_1\\\vdots\\f_p\end{pmatrix}.$$
Let $I$ denote the ideal $\langle f_1,\ldots, f_p\rangle$. According
to Lemma \ref{lem:elimJac} (in Appendix), there exists a polynomial
$h_1\in k[\mathfrak a]$, such that if $h_1(\mathbf a)\neq 0$, then $\sqrt{\langle \mathsf{MaxMinors}(\psi\circ\varphi_{\mathbf a}(\mathfrak
  M))\rangle}=\langle \varphi_{\mathbf a}(f_1),\ldots,\varphi_{\mathbf a}(f_p)\rangle\cap
k'[x_0,\ldots,x_{p-q}]$.

One remarks that there also exists a polynomial $h_2\in k[\mathfrak
  a]$ such that if $h_2(\mathbf a)\neq 0$, then $\varphi_{\mathbf a}(I)$ is
0-dimensional (since $f_1,\ldots, f_p$ is a
generic affine bilinear system with $p$ equations and $p$ variables, see
Proposition \ref{equidim}). From Lemma \ref{lem:radical} (in Appendix), there exists a polynomial $h_3$
such that if $h_3(\mathbf a)\neq 0$, then $\varphi_{\mathbf a}(I)$ is
radical. From now on, we suppose that $h_1(\mathbf a) h_2(\mathbf a)
h_3(\mathbf a)\neq 0$. If $(w_0,\ldots,w_{p-q})\in Var(\langle
\mathsf{MaxMinors}(\psi\circ\varphi_{\mathbf a}(\mathfrak M))\rangle)$ (where $Var$ denotes the variety), then the
set of points in $Var(\varphi_{\mathbf a}(I))$ whose projection is
$(w_0,\ldots,w_{p-q})$ can be obtained by solving an affine linear
system. The set of solutions of this system is nonempty and finite
(since $\varphi_{\mathbf a}(I)$ is 0-dimensional), thus it contains a unique
element. So there is a bijection between $Var(\varphi_{\mathbf a}(I))$ and
$Var(\langle \mathsf{MaxMinors}(\psi\circ\varphi_{\mathbf a}(\mathfrak
M))\rangle)$. Since $\varphi_{\mathbf a}(I)$ is radical,
$$\deg(\varphi_{\mathbf a}(I))=\deg(\sqrt{\langle \mathsf{MaxMinors}(\psi\circ\varphi_{\mathbf a}(\mathfrak
M))\rangle}).$$
From Corollary \ref{coro:degaff}, this degree is $p\choose {q-1}$.
According to Lemma \ref{lemGB1},
$$\begin{array}{rcl}\deg(\sqrt{\langle \mathsf{MaxMinors}(\psi\circ\varphi_{\mathbf a}(\mathfrak
M))\rangle})&\leq& \deg(\langle \mathsf{MaxMinors}(\psi\circ\varphi_{\mathbf a}(\mathfrak
M))\rangle)\\&\leq& \deg(\langle\textsf{Monomials}_{p-q}(q)\rangle)={p\choose {q-1}}.\end{array}$$ 
Therefore,
$$\deg(\sqrt{\langle \mathsf{MaxMinors}(\psi\circ\varphi_{\mathbf a}(\mathfrak
M))\rangle}) = \deg(\langle \mathsf{MaxMinors}(\psi\circ\varphi_{\mathbf a}(\mathfrak
M))\rangle)$$
and thus
$$\sqrt{\langle \mathsf{MaxMinors}(\psi\circ\varphi_{\mathbf a}(\mathfrak
M))\rangle} = \langle \mathsf{MaxMinors}(\psi\circ\varphi_{\mathbf a}(\mathfrak
M))\rangle.$$
Furthermore, the inequation $h_1(\mathbf a) h_2(\mathbf a) h_3(\mathbf a)\neq 0$ defines the wanted Zariski open set.
\hfill $\square$

\bigskip

\noindent{\em Proof of Lemma \ref{lemGB2}.}  Consider the Zariski open set
$O''=O'\cap O'''$ (where $O'$ is defined in Lemma \ref{lemGB1} and
$O'''$ is defined in Lemma \ref{lem:degproj}) and let $\mathbf a$ be
taken in $O''$. According to Lemma \ref{lemGB1},
$$\mathsf{Monomials}_{p-q}(q)\subset \LM(\langle
\mathsf{MaxMinors}(\psi\circ\varphi_{\mathbf a}(\mathfrak
M))\rangle).$$ A basis of $k'[x_0,\ldots, x_{p-q}]/\langle
\mathsf{Monomials}_{p-q}(q)\rangle$ is given by the set of all
monomials of degree less than $q$. Therefore, the dimension
of $k'[x_0,\ldots, x_{p-q}]/\langle
\mathsf{Monomials}_{p-q}(q)\rangle$ (as a $k'$-vector space) is
$p\choose {q-1}$.  Thus, from Lemma \ref{lem:degproj},
$$\deg(\langle \mathsf{MaxMinors}(\psi\circ\varphi_{\mathbf a}(\mathfrak
M))\rangle)={p\choose {q-1}}=\deg(\langle
\mathsf{Monomials}_{p-q}(q)\rangle).$$ Therefore, all polynomials in
$\langle \mathsf{MaxMinors}(\psi\circ\varphi_{\mathbf a}(\mathfrak M))\rangle$
have degree at least $q$.

Now let $g\neq 0$ be a polynomial in
$\langle\mathsf{MaxMinors}(\varphi_{\mathbf a}(\mathfrak M))\rangle$. Then there
exists $\mathbf v=(v_1,\ldots, v_{n_x-p+q})$ such that the specialized
polynomial verifies $\psi_{\mathbf{v}}(g)\neq 0$ and such that $\deg(\langle \mathsf{MaxMinors}(\psi_{\mathbf v}\circ\varphi_{\mathbf a}(\mathfrak
M))\rangle)={p\choose {q-1}}$.  Thus
$\psi_{\mathbf{v}}(g)$ is a polynomial of degree at least $q$ in
$k[x_0,\ldots, x_{p-q}]$. Now suppose by contradiction that
$\LM(g)\notin \langle \mathsf{Monomials}_{p-q}(q)\rangle$. Since
$\deg(\psi_{\mathbf v}(g))\geq q$, there exists a monomial $\mathfrak
m$ in $g$ such that $\mathfrak m\in \langle
\mathsf{Monomials}_{p-q}(q)\rangle$. Thus consider $g_1=g-\lambda\mathfrak
m+\lambda\NF(\mathfrak m)$. One remarks that $\LM(g)=\LM(g_1)\notin \langle
\mathsf{Monomials}_{p-q}(q)\rangle$. Since $g_1\in
\langle\mathsf{MaxMinors}(\varphi_{\mathbf a}(\mathfrak M))\rangle$, by a similar
argument there also exists a monomial $\mathfrak m_1\in \langle
\mathsf{Monomials}_{p-q}(q)\rangle$ in $g_1$. By induction construct
the sequence $g_i=g_{i-1}-\lambda_{i-1}\mathfrak m_{i-1}+\lambda_{i-1}\NF(\mathfrak
m_{i-1})$. This sequence is infinite and strictly decreasing (for the
induced partial ordering on polynomials: $h_1\prec h_2$ if
$\LM(h_1)\prec \LM(h_2)$ or if $\LM(h_1) = \LM(h_2)$ and
$h_1-\LM(h_1)\prec h_2-\LM(h_2)$). But, when $\prec$ is the grevlex
ordering, there does not exist such an infinite and strictly decreasing
sequence.

Therefore $\LM(g)\in\langle \mathsf{Monomials}_{p-q}(q)\rangle$, which concludes the proof.
\hfill$\square$

{\em Proof of Lemma \ref{lemGB3}.} In order to prove that the Zariski open set $O'\cap O''$ is nonempty, we exhibit an explicit element.
Consider the matrix ${\sf M}$ of ${\sf Mat}_\LL(p,q)$ whose
$(i,j)$-entry is $x_{i+j-2}$ if $0\leq i+j-2\leq p-q$ and $i\geq j$, else it is $0$.

%% $$M=\begin{pmatrix}
%% x_0&x_1&\dots&x_{p-q}&0&\dots&0\\
%% 0&x_0&\dots&\ddots&x_{p-q}&\ddots&0\\
%% \vdots&\ddots&\ddots&\ddots&\ddots&\ddots&\vdots\\
%% 0&0&\dots&\dots&\dots&x_{p-q-1}&x_{p-q}
%% \end{pmatrix}.$$
$$\mathsf{M}=\begin{pmatrix}
x_0&0&\dots&0\\
x_1&x_0&\ddots&0\\
\vdots&x_1&\ddots&\vdots\\
x_{p-q}&\ddots&\ddots&\vdots\\
\vdots&\ddots&\ddots&x_{p-q-1}\\
0&0&\dots&x_{p-q}
\end{pmatrix}.$$

Remark that ${\sf MaxMinors}({\sf M}))\subset
k[x_0,\ldots,x_{p-q}]$. Since $\langle {\sf
Monomials}_{p-q}(q)\rangle$ is a zero-dimensional ideal in
$k[x_0,\ldots,x_{p-q}]$, the fact that $\LM({\sf MaxMinors}({\sf
  M}))={\sf Monomials}_{p-q}(q)$ implies the equality of the monomial ideals  $\LM(\langle
{\sf MaxMinors}({\sf M})\rangle)=\langle {\sf
  Monomials}_{p-q}(q)\rangle$.  Thus, we prove in the sequel that
$\LM({\sf MaxMinors}({\sf M}))={\sf Monomials}_{p-q}(q)$.  

A first observation is that the cardinality of ${\sf MaxMinors}({\sf
  M})$ equals the cardinality of ${\sf Monomials}_{p-q}(q)$.  Let $m$ be a maximal minor of ${\sf
  M}$. Thus $m$ is the determinant of a $q\times q$ submatrix ${\sf
  M}'$ obtained by removing $p-q$ rows from ${\sf
  M}$. Let $i_1,\ldots,i_{p-q}$ be the indices of these rows (with
$i_1<\ldots <i_{p-q}$).  Denote by $\star$ the product coefficient by
coefficient of two matrices (i.e. the \emph{Hadamard product}) and let $\mathfrak{S}_{q}$ be the set of
$q\times q$ permutation matrices.  Thus
$m=\sum_{\sigma\in\mathfrak{S}_{q}} (-1)^{\mathsf{sgn}(\sigma)}\det(\sigma\star{\sf M}')$.

Since for all $ \sigma\in {\frak S}_q$, $ \det(\sigma\star {\sf M}')$
is a monomial, there exists $ \sigma^0\in {\frak S}_q$ such that 
$\LM(m)=\pm \det(\sigma^0\star {\sf M}')$.

We prove now that $\sigma^0=\mathsf{id}$.  Suppose
by contradiction that $\sigma^0\neq \mathsf{id}$. In the sequel, we denote by 
\begin{itemize}
\item ${\sf M}'[i,j]$ the $(i,j)$-entry of ${\sf M}'$. 
\item $\mathbf{e}_i$ the $q\times 1$ unit vector whose $i$-th coordinate is $1$ and all its other coordinates are $0$; 
\item $\sigma^0_j$ is the integer $i$ such that $\sigma^0\mathbf{e}_j=\mathbf{e}_i$.
\end{itemize}
Since, by assumption, $\sigma^0\neq \mathsf{id}$, there exists $1\leq i<j\leq
q$ such that $\sigma^0_j>\sigma^0_i$. Because of the structure of ${\sf M}$, we know that for
the $grevlex$ ordering $x_0\succ \cdots\succ x_{n_x}$, $${\sf
  M}'[i,\sigma^0_j] {\sf M}'[j,\sigma^0_i]\succ{\sf M}'[i,\sigma^0_i] {\sf
  M}'[j,\sigma^0_j].$$ Let $\sigma'$ be defined by
$$\sigma'_k=\begin{cases}
\sigma^0_k\text{ if }k\neq i\text{ and }k\neq j\\
\sigma^0_j\text{ if }k=i\\
\sigma^0_i\text{ if }k=j
\end{cases}$$ Then $\det(\sigma'\star {\sf M}')\succ\det(\sigma^0\star
{\sf M}')$ and by induction $\det(\mathsf{id}\star {\sf M}')\succ\det(\sigma^0\star
{\sf M}')$. This also proves that the coefficient of $\det(\mathsf{id}\star
{\sf M}')$ in ${\sf MaxMinors}({\sf M})$ is $1$ and contradicts the
fact that $\LM(m)=\pm\det(\sigma^0\star {\sf M}')$.

This proved that $\LM(m)=|\det(\mathsf{id}\star {\sf M}')|$.  Now one can
remark that
$$\det(\mathsf{id}\star {\sf M}')=x_0^{i_1-1} x_1^{i_2-i_1-1} x_2^{i_3-i_2-1} \ldots x_{p-q}^{p-i_{p-q}-1}.$$
If $m_1,m_2$ are distinct elements in ${\sf MaxMinors}({\sf M})$,
then $\LM(m_1)\neq \LM(m_2)$.  For all $m$ in ${\sf
  MaxMinors}({\sf M})$, $\LM(m)\in {\sf Monomials}_{p-q}(q)$, and
${\sf MaxMinors}({\sf M})$ has the same cardinality as ${\sf
  Monomials}_{p-q}(q)$. Therefore, one can deduce that $\LM({\sf MaxMinors}({\sf
  M}))={\sf Monomials}_{p-q}(q)$.  \hfill$\square$

\subsection{An extension of the $F_5$ criterion for bilinear systems}
\label{sectionalgo}
We can now present the main algorithm of this section. Given a
sequence of homogeneous bilinear forms $F=(f_1, \ldots, f_m)\subset R$
generating an ideal $I\subset R$, $\prec$ the $grevlex$ monomial
ordering on $R$ with $x_0\succ \cdots x_{n_x}\succ y_0\succ \cdots
y_{n_y}$, it returns a set of pairs $(g, f_i)$ such that $g\in
I_{i-1}:f_i$ and $g\notin I_{i-1}$ (for $i>\min(n_x+1, n_y+1)$).  Following Theorem
\ref{correct} and \ref{minorsGB}, this is done by considering the
matrices $\jac_{\x}(F_i)$ (resp. $\jac_{\y}(F_i)$) for $i>n_x+1$
(resp. $i>n_y+1$) and performing a row echelon form on ${\sf
  Macaulay}_\prec({\sf MaxMinors}(\jac_{\x}(F_i)),n_x+1)$ (resp. ${\sf
  Macaulay}_\prec({\sf MaxMinors}(\jac_{\y}(F_i)),n_y+1)$).

First we describe the subroutine \textsf{{\bf Reduce}} (Algorithm \ref{algo:reduce}) which reduces a set of homogeneous polynomials of the same degree:
\begin{alg}\label{algo:reduce}
\textsf{{\bf Reduce}}
\begin{algorithmic}[1]
\Require $(S,q)$ where $S$ is a set of homogeneous polynomials of degree $q$.
\Ensure $T$ is a reduced set of homogeneous polynomials of degree $q$. 
\State $\mathsf M\leftarrow \mathsf{Macaulay}(S,q)$.
\State $\mathsf M\leftarrow \mathsf{RowEchelonForm}(\mathsf M)$.
\State Return $T$ the set of polynomials corresponding to the rows of $\mathsf M$.
\end{algorithmic}
\end{alg}

The main algorithm uses this subroutine in order to compute a row echelon form of the matrix\\${\sf
  Macaulay}_\prec({\sf MaxMinors}(\jac_{\x}(F_i)),n_x+1)$ (resp. ${\sf
  Macaulay}_\prec({\sf MaxMinors}(\jac_{\y}(F_i)),n_y+1)$):
\begin{alg}\label{algoprinc}
\textsf{{\bf BLcriterion}}
\begin{algorithmic}[1]
\Require $\begin{cases}m\text{ bilinear polynomials }f_1,\ldots,f_m\text{ such that }m\leq n_x+n_y.\\
< \text{ a monomial ordering over }k[x_0,\ldots,x_{n_x},y_0,\ldots,y_{n_y}]\end{cases}$
\Ensure $V$ a set of pairs $(h,f_i)$ such that $h\in I_{i-1}:f_i$.
\State $V\leftarrow \emptyset$
\For{$i$ from $2$ to $m$}
\If{$i>n_y+1$}
\State $T\leftarrow \mathsf{{\bf Reduce}}({\sf MaxMinors}(\jac_{\y}(F_{i-1})),n_y+1)$.
\For{$h$ in $T$}
\State $V\leftarrow V\cup \{(h,f_i)\}$
\EndFor
\EndIf
\If{$i>n_x+1$}
\State $T'\leftarrow \mathsf{{\bf Reduce}}({\sf MaxMinors}(\jac_{\x}(F_{i-1})),n_x+1)$.
\For{$h$ in $T'$}
\State $V\leftarrow V\cup \{(h,f_i)\}$
\EndFor
\EndIf
\EndFor
\State Return $V$
\end{algorithmic}
\end{alg}

The following Proposition explains how the output of Algorithm
\ref{algoprinc} is related to reductions to zero occurring during the Matrix
$F_5$ Algorithm.
\begin{prop}[Extended $F_5$ criterion for bilinear systems]
Let $f_1,\ldots,f_m$ be bilinear polynomials and $\prec$ be a monomial ordering.
Let $(t,f_i)$ be the signature of a row during the Matrix $F_5$ Algorithm and let $V$ be the output of Algorithm \textsc{BLcriterion}. Then if there exists $(h,f_i)$ in $V$ such that $LM(h)=t$, then the row with signature $(t,f_i)$ will be reduced to zero.
\end{prop}
\noindent{\em Proof.}
According to Theorem \ref{correct}, $h f_i\in I_{i-1}$. Therefore
$$t f_i=(h-t) f_i +\sum_{j=1}^{i-1} g_j f_j.$$
This implies that the  row with signature $(t,f_i)$ is a linear combination of preceding rows in the matrix $\mathsf{Macaulay}(F_i, \deg(t f_i))$. Hence this row will be reduced to zero.
\hfill $\square$

\bigskip

Now we can merge this extended criterion with the Matrix $F_5$
Algorithm. To do so, we denote by $V$ the output of
\textsc{BLcriterion} ($V$ has to be computed at the beginning of
Matrix $F_5$ Algorithm), and we replace in Algorithm
\ref{matrixF5} the $F_5$\textsc{criterion} by the following
\textsc{Bilin}$F_5$\textsc{criterion}:
\begin{alg}
\textsc{Bilin}$F_5$\textsc{criterion} - returns a boolean
\begin{algorithmic}[1]

\Require $\begin{cases}
(t,f_i) \text{ the signature of a row}\\
\text{A matrix }\mathcal{M}\text{ in row echelon form}
\end{cases}$

\State Return $\begin{cases}t \text{ is the leading monomial of a row of }\mathcal{M}\textbf{ or }\\\exists (h,f_i)\in V
\text{ such that }\LM(h)=t\end{cases}$
\end{algorithmic}
\end{alg}

\section{$F_5$ without reduction to zero for generic bilinear systems}
\label{sec:gener}
\subsection{Main results}
The goal of this part of the paper is to show that Algorithm
\ref{algoprinc} finds all reductions to zero for generic bilinear
systems. In order to describe the structure of ideals generated by
generic bilinear systems, we define a notion of \emph{bi-regularity}
(Definition \ref{defbireg}). For bi-regular systems, we give a complete
description of the syzygy module (Proposition \ref{propxy} and Corollary \ref{corointro}). Finally, we show that, for
such systems, Algorithm \ref{algoprinc} finds all reductions to
zero and that generic bilinear systems are bi-regular
(Theorem \ref{genericity}), assuming a conjecture about the kernel of
generic matrices whose entries are linear forms (Conjecture
\ref{conjec}).

\subsection{Kernel of matrices whose entries are linear forms}

Consider an monomial ordering $\prec$ such that
its restriction to $k[x_0,\ldots,x_{n_x}]$
(resp. $k[y_0,\ldots,y_{n_y}]$) is the $grevlex$ ordering (for
instance the usual \emph{grevlex} ordering with
$x_0\succ x_1\succ \ldots\succ y_0\succ\ldots\succ y_{n_y}$).

Let $\ell,c,n_x$ be integers such that $c<\ell\leq n_x+c-1$. Let
${\cal M}$ be the set of matrices $\ell\times c$ where coefficients are
linear forms of $k[x_0,\ldots,x_{n_x}]$. Let ${\cal T}$ be the set of
$\ell\times (\ell-c-1)$ matrices $\mathsf T$ such that:
\begin{itemize}
\item each column of $\mathsf T$ has exactly one $1$ and the rest of the coefficients are $0$.
\item each row of $\mathsf T$ has at most one $1$ and all the other coefficients are $0$.
\item $(\mathsf T[i_1,j_1]=\mathsf T[i_2,j_2]=1$ and $i_1<i_2)\Rightarrow j_1<j_2$
\end{itemize}
If $\mathsf{T}\in {\cal T}$ and $\mathsf M\in{\cal M}$, we denote by $\mathsf M_{\mathsf T}$ the $\ell\times(\ell-1)$ matrix obtained by adding to $\mathsf M$ the columns of $\mathsf T$.
According to the proof of Lemma \ref{lemmemi}, some elements of the left kernel of a matrix $\mathsf M$ can be expressed as vectors of maximal minors:
$$\forall \mathsf T\in\mathcal{T},
\begin{pmatrix}
\mathsf{minor}(\mathsf M_{\mathsf T},1)\\
-\mathsf{minor}(\mathsf M_{\mathsf T},2)\\
\vdots\\
(-1)^{m+1} \mathsf{minor}(\mathsf M_{\mathsf T},m)
\end{pmatrix}\in \Ker_L(\mathsf M).$$

Actually, we observed experimentally that kernels of 
random matrices $\mathsf M\in\mathcal{M}$ are generated by those vectors of
minors. This leads to the formulation of the following conjecture:
\begin{conj}\label{conjec}
The set of matrices $\mathsf M\in \mathcal{M}$ such that
$$\Ker_L(\mathsf M)=\left\langle\left\{\begin{pmatrix}
\mathsf{minor}(\mathsf M_{\mathsf T},1)\\
-\mathsf{minor}(\mathsf M_{\mathsf T},2)\\
\vdots\\
(-1)^{m+1} \mathsf{minor}(\mathsf M_{\mathsf T},m)
\end{pmatrix}\right\}_{\mathsf T\in\mathcal{T}}\right\rangle$$
contains a nonempty Zariski open subset of $\mathcal{M}$.
\end{conj}

\subsection{Structure of generic bilinear systems}

With the following definition, we try to give an analog of regular
sequences for bilinear systems. This definition is closely
related to the generic behaviour of Algorithm \ref{algoprinc}.

\begin{rem}
  In the following, $\mathsf{Monomials}^\x_n(d)$
  (resp. $\mathsf{Monomials}^\y_n(d)$) denotes the set of monomials of
  degree $d$ in $k[x_0,\ldots, x_n]$ (resp. $k[y_0,\ldots, y_n]$). If
  $n<0$, we use the convention
  $\mathsf{Monomials}^\x_n(d)=\mathsf{Monomials}^\y_n(d)=\emptyset$.
\end{rem}

\begin{defn}\label{defbireg}
Let $m\leq n_x+n_y$ and $f_1,\ldots,f_m$ be bilinear polynomials of $R$. We
say that the polynomial sequence $(f_1,\ldots,f_m)$ is a \emph{bi-regular sequence} if $m=1$ or if $(f_1,\ldots, f_{m-1})$ is a bi-regular sequence and
\begin{eqnarray*}
\LM(I_{m-1}:f_m)&=
\langle\mathsf{Monomials}^\x_{m-n_y-2}(n_y+1)\rangle\\
&+ \langle\mathsf{Monomials}^\y_{m-n_x-2}(n_x+1)\rangle\\
&+ \LM(I_{m-1})
\end{eqnarray*}
\end{defn}

In the following, we use the notations:
\begin{itemize}
\item ${\cal BL}(n_x, n_y)$ the $k$-vector space of bilinear
  polynomials in $K[x_0, \ldots, x_{n_x}, y_0, \ldots, y_{n_y}]$;
\item $X\subset k[x_0, \ldots, x_{n_x}, y_0, \ldots, y_{n_y}]$
  (resp. $Y$) is the ideal $\langle x_0, \ldots, x_{n_x}\rangle$
  (resp. $\langle y_0, \ldots, y_{n_y}\rangle$);
\item An ideal is called \emph{bihomogeneous} if there exists a set of bihomogeneous generators. In particular, ideals spanned by bilinear polynomials are bihomogeneous.
\item $J_i$ denotes the saturated ideal $I_i : (X \cap Y)^\infty$;
\item Given a polynomial sequence $(f_1,\ldots, f_m)$, we denote by
  $Syz_{triv}$ the module of trivial syzygies, i.e. the set of all 
  syzygies $(s_1,\ldots, s_m)$ such that $$\forall i, s_i\in \langle
  f_1,\ldots, f_{i-1},f_{i+1},\ldots, f_m\rangle;$$
\item A primary ideal $P\subset R$ is called \emph{admissible} if
  $\langle x_0,\ldots, x_{n_x}\rangle \not\subset\sqrt{P}$ and
  $\langle y_0,\ldots, y_{n_y}\rangle \not\subset\sqrt{P}$;
\item Let $E$ be a $k$-vector space such that $\dim(E)<\infty$. We say that a property $\mathcal P$ is \emph{generic} if it is satisfied on a nonempty open subset of $E$ (for the Zariski topology), i.e. $\exists h\in k[\mathfrak a_1,\ldots,\mathfrak a_{\dim(E)}], h\neq 0$, such that 
$$\mathcal P\text{ is does not hold on }(a_1,\ldots,a_{\dim(E)})\Rightarrow h(a_1,\ldots,a_{\dim(E)})=0.$$
\end{itemize}
Without loss of generality, we suppose in the sequel that $n_x\leq
n_y$.

\begin{lem}\label{lemNonAdm}
Let $I_m$ be an ideal spanned by $m$ generic bilinear equations $f_1,\ldots, f_m$ and $I_m=\cap_{P\in\mathcal{P}}P$ be a minimal primary decomposition. Let $P_0\in\mathcal{P}$ be one of its primary non-admissible components. If $m<n_x+1$ (resp. $m< n_y+1$), then $X \not\subset \sqrt P_0$ (resp. $Y \not\subset \sqrt P_0$).
\end{lem}

\noindent{\em Proof.}
Suppose that $m<n_x+1$. Consider the field $k'=k(y_0,\ldots,y_{n_y})$ and the canonical inclusion
$$\psi : R \rightarrow k'[x_0,\ldots,x_{n_x}].$$
    $\psi(I_m)$ is an ideal of $k'[x_0,\ldots,x_{n_x}]$ spanned by
    $m$ polynomials of $k'[x_0,\ldots,x_{n_x}]$. Generically, the system 
    $(\psi(f_1),\ldots,\psi(f_m))$ is a regular sequence of
    $k'[x_0,\ldots,x_{n_x}]$.  Thus there exists an polynomial
    $f\in X$ (homogeneous in the $x_i$s) such that $\psi(f)$ is not a
    divisor of $0$ in $k'[x_0,\ldots,x_{n_x}]/\psi(I_m)$. This
    means that $\psi(I_m) : \psi(f)=\psi(I_m)$. Suppose the
    assertion of Lemma \ref{lemNonAdm} is false. Then $X\subset \sqrt{P_0}$ and hence, $f\in \sqrt{P_0}$. Therefore there exists
    $g\in k[y_0,\ldots,y_{n_y}]$ such that, in $R$, $g f \in
    \sqrt{I_m}$ (take $g$ in
    $(\cap_{P\in\mathcal{P}\setminus\{P_0\}}\sqrt{P})\setminus\{\sqrt{P_0}\}$
    which is nonempty). Thus
    $\psi(f)\in\sqrt{\psi(I_m)}$ (since $\psi(g)$ is invertible in
    $k'$), which is impossible since $\psi(I_m) : \psi(f)=\psi(I_m)$.
\hfill $\square$

\begin{lem}\label{lem:important}
    \begin{itemize}
    \item If $m\leq n_x$ there exists a nonempty Zariski-open set
      ${\cal O}\subset {\cal BL}_K(n_x, n_y)^m$ such that $(f_1,
      \ldots, f_m)\subset {\cal O}$ implies that $I_m$ has co-dimension
      $m$ and all the components of a minimal primary decomposition of
      $I_m$ are admissible;
    \item if $n_x+1\leq m$, then there exists a nonempty
      Zariski-open set ${\cal O}\subset {\cal BL}_K(n_x, n_y)^m$ such
      that $(f_1, \ldots, f_m)\subset {\cal O}$ implies that $X$ is a
      prime associated to $\sqrt{I_m}$;
    \item if $n_y+1\leq m$, then there exists a nonempty Zariski-open
      set ${\cal O}\subset {\cal BL}_K(n_x, n_y)^m$ such that $(f_1,
      \ldots, f_m)\subset {\cal O}$ implies that $Y$ is a prime
      associated to $\sqrt{I_m}$.
    \end{itemize}
  \end{lem}
  \noindent{\em Proof.}
     \begin{itemize}
    \item If $m\leq n_x$, then by Lemma \ref{lemNonAdm},
      $J_m=I_m$. Then according to Theorem \ref{generbireg}, there
      exists a nonempty Zariski-open set ${\cal O}\subset {\cal
        BL}_K(n_x, n_y)^m$ such that $(f_1, \ldots, f_m)\subset {\cal
        O}$ implies that $(f_1,\ldots, f_m)$ is a regular
      sequence. Therefore, $I_m$ has co-dimension $m$ and all the
      components of a minimal primary decomposition of $I_m$ are
      admissible.
    \item If $n_x+1\leq m$, then according to Proposition
      \ref{equidim}, $J_m=(I_m:Y^\infty):X^\infty$ is equidimensional
      of co-dimension $m$. Let $V_x$ be the set
      $\{(0,\ldots,0,a_0,\ldots,a_{n_y}) | a_i\in k\}$. Since
      $V_x\subset Var(I_m:Y^\infty)$ and $\mathsf{codim}(V_x)=n_x+1$, it can be
      deduced that $V_x \not\subset Var(J_m)$ and
      $Var(I_m:Y^\infty)=Var(J_m)\cup V_x$. This means that
      $\sqrt{I_m:Y^\infty}=\sqrt{J_m}\cap X$ and $\sqrt{J_m}\not
      \subset X$. Thus $X$ is a prime associated to
      $\sqrt{I_m:Y^\infty}$. Since $Y$ is not a subset of $X$, $X$ is
      also a prime ideal associated to $\sqrt{I_m}$.
    \item Similar proof in the case $n_y+1\leq m$.
  \end{itemize}
  \hfill $\square$
  \begin{lem}\label{lem:localring}
    Suppose that the local ring $R_X/I_X$ (resp. $R_Y/I_Y$) is regular
    and that $X$ (resp. $Y$) is a prime ideal associated to $\sqrt{I}$
    and let $Q$ be an isolated primary component of
    a minimal primary decomposition of $I$ containing $X$
    (resp. $Y$). Then $Q=X$ (resp. $Q=Y$).
  \end{lem}
 
  \noindent{\em Proof.}
    By assumption, $X$ is a prime ideal associated to
    $\sqrt{I}$. Then, there exists an isolated primary component of a
    minimal primary decomposition of $I$ which contains a power of $X$ and does
    not meet $R\setminus X$. This proves that $I_X$ does not
    contain a unit in $R_X$.

    By assumption $R_X/I_X$ is regular and local, then $R_X/I_X$ is an
    integral ring (see e.g. \cite[Corollary 10.14]{eisenbud}) which
    implies that $I_X$ is prime and does not contain a unit in $R_X$.

    Let $I=Q_1\cap\dots\cap Q_s$ be a minimal primary decomposition of
    $I$. In the sequel, $Q_{i_X}$ denotes the localization of $Q_i$ by
    $X$. Suppose first that there exists $1\leq i \leq s$ such that
    $I_X=Q_{i_X}$ with $Q_{i}$ non-admissible which does not meet the
    multiplicatively closed part $R\setminus X$ .  Then $Q_{i_X}$ is
    obviously prime which implies that $Q_i$ itself is prime
    \cite[Proposition 3.11 (iv)]{AtMc}. Our claim follows.

    \medskip It remains to prove that $I_X=Q_{i_X}$ for some $1\leq i
    \leq s$. Suppose that the $Q_i$'s are numbered such that $Q_j$
    meets the multiplicatively closed set $R\setminus X$ for $r+1\leq
    j \leq s$ but not $Q_1, \ldots, Q_r$. $I_X=Q_{1_X}\cap \dots\cap
    Q_{r_X}$ and it is a minimal primary decomposition
    \cite[Proposition 4.9]{AtMc}. Hence, since $I_X$ is prime, $r=1$
    and $Q_1$ is the isolated minimal primary component containing
    $X$.

    Proving that $Q=Y$ in the case where $R_Y/I_Y$ is regular and that
    $Y$ is a prime associated to $\sqrt{I}$ is done in the same way.
  \hfill $\square$

\begin{prop}\label{prop:nonadm}
  Let $k$ be a field of characteristic $0$. There exists a nonempty
  Zariski-open set ${\cal O}\subset {\cal BL}(n_x, n_y)^m$ such that
  for all $(f_1, \ldots, f_m)\subset {\cal O}$ the non-admissible
  components of a minimal primary decomposition of $\langle f_1,
  \ldots, f_m\rangle$ are either $X$ or $Y$.
\end{prop}

\noindent{\em Proof.}
Suppose that $n_x+1\leq m$. Then, from Lemma \ref{lem:important},
there exists a nonempty Zariski-open set $O_1$ such that $X$ is an
associated prime to $\sqrt{{ I}}$. Note also that this implies that
${I}_X$ has co-dimension $n_x+1$. Thus, {f}rom Lemma
\ref{lem:localring}, it is sufficient to prove that there exists a
nonempty Zariski-open set $O_2$ such that for all $(f_1, \ldots,
f_m)\in O_1\cap O_2$, $R_X/I_X$ is a regular local ring.

{F}rom the Jacobian Criterion (see e.g. \cite{eisenbud}, Theorem
16.19), the local ring $R_X/ I_X$ is regular if and only if
${\rm jac}({f}_1, \ldots, { f}_m)$ taken modulo $X$ has
co-dimension $n_x+1$. Since the generators of ${I}$ are
bilinear, the latter condition is equivalent to saying that the matrix
$$
J_X=\left [\begin{array}{ccc}
  \frac{\partial f_1}{\partial x_0} & \cdots &   \frac{\partial f_1}{\partial x_{n_x}}\\
  \vdots & \cdots & \vdots\\
  \frac{\partial f_m}{\partial x_0} & \cdots &   \frac{\partial f_m}{\partial x_{n_x}}\\
\end{array}\right ]
$$
has rank $n_x+1$. We prove below that there exists a nonempty
Zariski-open set $O_3$ such that for all $(f_1, \ldots, f_m)\in
O_3$, $J_X$ has rank $n_x+1$.

Let ${\frak c}_1, \ldots, {\frak c}_m$ be vectors of coordinates of
${\cal BL}(n_x, n_y)^m$, ${\frak {M}}$ be the vector of all bilinear
monomials in $R$ with respect to the partition $[x_0, \ldots,
x_{n_x}], [y_0, \ldots, y_{n_y}]$ and ${\frak K}$ be the field of
rational fractions $k({\frak c}_1, \ldots, {\frak c}_m)$. Consider the
polynomials ${\frak f}_i={\frak M}.{\frak c}_i^T$ for $1\leq i\leq m$
and the Zariski-open set $O_3$ in ${\cal BL}(n_x, n_y)^m$ defined by
the non-vanishing of all the coefficients of the maximal minors of the
matrix
$$
{\frak J}_X=\left [\begin{array}{ccc}
    \frac{\partial {\frak f}_1}{\partial x_0} & \cdots &   \frac{\partial {\frak f_1}}{\partial x_{n_x}}\\
    \vdots & \cdots & \vdots\\
    \frac{\partial {\frak f}_m}{\partial x_0} & \cdots &   \frac{\partial {\frak f}_m}{\partial x_{n_x}}\\
\end{array}\right ].
$$
It is obvious that $(f_1, \ldots, f_m)\in O_3$ implies that $J_X$ has
rank $n_x+1$; our claim follows.

\medskip In the case where $n_y\leq m$. The proof follows the same
pattern using Lemmas \ref{lem:important} and \ref{lem:localring} and
the Jacobian criterion. The only difference is that one has to prove
that there exists a nonempty Zariski-open set $O_4$ such that for all $(f_1, \ldots, f_m)\in O_4$ the matrix 
$$
J_Y=\left [\begin{array}{ccc}
  \frac{\partial f_1}{\partial y_0} & \cdots &   \frac{\partial f_1}{\partial y_{n_x}}\\
  \vdots & \cdots & \vdots\\
  \frac{\partial f_m}{\partial y_0} & \cdots &   \frac{\partial f_m}{\partial y_{n_y}}\\
\end{array}\right ]
$$
has rank $n_y+1$, which is done as above.
\hfill $\square$

\begin{rem}
  The proof of Proposition \ref{prop:nonadm} relies on the use of the
  Jacobian Criterion. {F}rom \cite[Theorem 16.19]{eisenbud}, it
  remains valid if the characteristic of $k$ is large enough so that
  the residue class field of $X$ (resp. $Y$) is separable. 
\end{rem}

The two following propositions explain why the rows reduced to zero
in the generic case during the $F_5$ Algorithm have a signature
$(t,f_i)$ such that $t\in k[x_0,\ldots,x_{n_x}]$ or $t\in
k[y_0,\ldots,y_{n_y}]$.

\begin{prop}\label{propxy}
Let $m$ be an integer such that $m\leq n_x+n_y$. Let $L$ be the set of
bilinear systems with $m$ polynomials ($L\subset R^m$). Then the set
of bilinear systems $f_1,\ldots,f_m$ such that $Syz=\langle (Syz\cap
k[x_0,\ldots,x_{n_x}]^m)\cup (Syz\cap k[y_0,\ldots,y_{n_y}]^m) \cup
Syz_{triv}\rangle$ is a nonempty Zariski-open subset of $L$.
\end{prop}

\noindent{\em Proof.}
  Let $s=(s_1,\ldots,s_m)$ be a syzygy. Thus, $s_m$ is in $I_{m-1} : f_m$. We can suppose without loss of generality that the $s_i$ are bihomogeneous of same bi-degree (Proposition \ref{prop:bihom}). According to Theorem \ref{generbireg}, there exists a nonempty Zariski open set $O_1\subset {\cal BL}(n_x, n_y)^m$, such that if $(f_1,\ldots,f_m)\in O_1$, then $f_m$ is not a divisor of $0$ in $R/J_{m-1}$. 
  We can deduce from this observation that $s_m\in J_{m-1}$. So
  $s_m\in I_{m-1}$ or there exists $P$ a non-admissible primary
  component of
  $I_{m-1}$ such that $s_m\notin P$. Assume that $s_m\notin I_{m-1}$. From Proposition \ref{prop:nonadm}, there exists a nonempty Zariski open set $O_2\subset {\cal BL}(n_x, n_y)^m$, such that if $(f_1,\ldots,f_m)\in O_2$, then $\langle x_0,\ldots,x_{n_x} \rangle = P$
  (or $\langle y_0,\ldots,y_{n_y} = P$). This means
  that, generically, $s_m\in k[y_0,\ldots, y_{n_y}]$ (or $s_m\in
  k[x_0,\ldots, x_{n_x}]$).

Finally, we see that, if $(f_1,\ldots,f_m)\in O_1\cap O_2$, then $s_m\in I_{m-1}\cup k[y_0,\ldots,
  y_{n_y}] \cup k[x_0,\ldots, x_{n_x}]$. Since the syzygy module of a
bihomogeneous system is generated by bihomogeneous syzygies, it can be
deduced that $Syz=\langle (Syz\cap k[x_0,\ldots,x_{n_x}]^m)\cup
(Syz\cap k[y_0,\ldots,y_{n_y}]^m) \cup Syz_{triv}\rangle$.
\hfill $\square$

\begin{prop}\label{aff}
Let $V$ be the output of Algorithm \textsc{BLcriterion} and let $(h,f_i)$ be an element of $V$. Then
\begin{itemize}
\item if $h\in k[x_0,\ldots,x_{n_x}]$, then $\forall j, y_j h\in I_{i-1}$.
\item if $h\in k[y_0,\ldots,y_{n_y}]$, then $\forall j, x_j h\in I_{i-1}$.
\end{itemize}
\end{prop}

\noindent{\em Proof.}
Suppose that $h\in
k[x_0,\ldots,x_{n_x}]$ is a maximal minor of $\jac_\y(F_{i-1})$ (the proof is similar if $h\in
k[y_0,\ldots,y_{n_y}]$). Consider the matrix $\jac_\x(F_{i-1})$ as defined in
Algorithm \ref{algoprinc}. Then there exists an $(i-1)\times
(i-1)$ extension $\mathsf M_T$ of $\jac_\x(F_{i-1})$ such that $\det(M_T)=h$ (similarly to the proof of Lemma \ref{lemmemi}). Let
$0\leq j\leq n_y$ be an integer.  Consider the polynomials
$h_1,\ldots,h_{i-1}$, where $h_k$ is the determinant of the
$(i-2)\times (i-2)$ matrix obtained by removing the $(j+1)$th column and
the $k$th row from $\mathsf M_T$.

Then we can remark that 
$$\begin{pmatrix}h_1&-h_2&\dots&(-1)^{i-1} h_{i-2}&(-1)^i h_{i-1}\end{pmatrix}\cdot\mathsf M_T=\begin{pmatrix}0&\dots&0&(-1)^j \det(\mathsf M_T)&0&\dots&0
\end{pmatrix}$$
where the only non-zero component is in the $(j+1)$th column.
Keeping only the $n_y+1$ first columns of $\mathsf M_T$, we obtain
$$\begin{pmatrix}h_1&-h_2&\dots&(-1)^{n_y} h_{n_y+1}\end{pmatrix}\cdot \jac_\x(F_{i-1})=\begin{pmatrix}0&\dots&0&(-1)^j\det(A_T)&0&\dots&0
\end{pmatrix}$$
Since $\jac_\x(\mathsf F_{i-1})\cdot \begin{pmatrix}y_0\\\vdots \\y_{n_y}\end{pmatrix} =\begin{pmatrix}f_1\\\vdots\\f_{i-1}\end{pmatrix}$, 
the following equality holds
$$\begin{pmatrix}h_1&-h_2&\dots&(-1)^{n_y-1} h_{n_y}&(-1)^{n_y} h_{n_y+1}
\end{pmatrix}\cdot \begin{pmatrix}f_1\\\vdots\\f_{i-1}\end{pmatrix}=y_j \det(\mathsf M_T)=y_j h.$$
This implies that $y_j h\in I_{i-1}$.
\hfill $\square$

\begin{cor} \label{corointro}
Let $m$ be an integer such that $m\leq n_x+n_y$ and let
$f_1,\ldots,f_m$ be bilinear polynomials. Let $V$ be the output of Algorithm \textsc{BLcriterion}. Assume that 
$$(I_{m-1}:f_m)\cap k[x_0,\ldots,x_{n_x}]=\langle \{h\in k[x_0,\ldots,x_{n_x}]~~:~~ (h,f_m)\in V\} \rangle. $$
$$(I_{m-1}:f_m)\cap k[y_0,\ldots,y_{n_y}]=\langle \{h\in k[y_0,\ldots,y_{n_y}]~~:~~(h,f_m)\in V\} \rangle. $$
Let $G_x$ (resp $G_y$) be a
Gröbner basis of $(I_{m-1}:f_m)\cap k[x_0,\ldots,x_{n_x}]$
(resp. $(I_{m-1}:f_m)\cap k[y_0,\ldots,y_{n_y}]$) and let $G_{m-1}$ be
a Gröbner basis of $I_{m-1}$.  If $Syz=\langle (Syz\cap
k[x_0,\ldots,x_{n_x}]^m)\cup (Syz\cap k[y_0,\ldots,y_{n_y}]^m) \cup
Syz_{triv}\rangle$, then $G_x\cup G_y\cup G_{m-1}$ is a Gröbner basis
of $I_{m-1}:f_m$.
\end{cor}

\noindent{\em Proof.}
  Let $f\in I_{m-1}:f_m$ be a polynomial. Thus there exist $s_1,\ldots, s_{m-1}$ such that $(s_1,\ldots, s_{m-1},f)\in Syz$. Since $I_{m-1}$ and $f_m$
  are bihomogeneous, we can suppose without loss of generality that
  $f$ is bihomogeneous (Proposition \ref{prop:bihom}). 
Let $(d_1,d_2)$ denote its
  bi-degree. 
\begin{itemize}
\item If $d_2=0$ (resp. $d_1=0$), then $f\in \langle G_x \rangle$ (resp. $f\in \langle G_y \rangle$).
\item Let $G_x=\{g^{(x)}_i\}_{1\leq i\leq \mathsf{card}(G_x)}$ and
  $G_y=\{g^{(y)}_i\}_{1\leq i\leq \mathsf{card}(G_y)}$. If $d_1\neq 0$
  and $d_2\neq 0$ then, since $Syz=\langle (Syz\cap
  k[x_0,\ldots,x_{n_x}]^m)\cup (Syz\cap k[y_0,\ldots,y_{n_y}]^m) \cup
  Syz_{triv}\rangle$,
$$f =
\sum_{1\leq i \leq \mathsf{card}(G_x)} q_i g^{(x)}_i + \sum_{1\leq i \leq \mathsf{card}(G_y)} q'_i
g^{(y)}_i + t$$ where $t\in I_{m-1}$ is a bihomogeneous polynomial and the $q_i$ and $q'_i$ are also
bihomogeneous. Since $d_2\neq 0$ and $g^{(x)}_i\in
k[x_0,\ldots,x_{n_x}]$, $q_i$ must be in $\langle y_0,\ldots,
y_{n_y}\rangle$. According to Proposition \ref{aff}, $\forall i, q_i g^{(x)}_i\in I_{m-1}$. By a similar argument, $\forall i, q'_i g^{(y)}_i\in
I_{m-1}$.  Finally, $f\in I_{m-1}$.
\end{itemize}
We just proved that $I_{m-1}:f_m=I_{m-1}\cup \langle G_x\rangle\cup
\langle G_y\rangle$. Thus, $G_x\cup G_y\cup G_{m-1}$ is a Gröbner
basis of $I_{m-1}:f_m$.
\hfill $\square$

Corollary \ref{corointro} shows that, when a bilinear system is bi-regular, it is
possible to find a Gröbner basis of $I_{m-1}:f_m$ (which yields the
monomials $t$ such that the row $(t,f_m)$ reduces to zero) as soon as
we know the three Gröbner bases $G_x$, $G_y$, and $G_{m-1}$. In fact,
we only need $G_x$ and $G_y$ since the reductions to zero
corresponding to $G_{m-1}$ are eliminated by the usual $F_5$
criterion. Fortunately, we can obtain $G_x$ and $G_y$ just by
performing linear algebra over the maximal minors of a matrix (Theorem \ref{minorsGB}). 

We now present the main result of this section.
If we suppose that Conjecture \ref{conjec} is true, then the following Theorem shows that generic bilinear systems are bi-regular.
\begin{thm}\label{genericity}
Let $m,n_x,n_y\in\mathbb{N}$ such that $m<n_x+n_y$. 
The set of bi-regular sequences $(f_1,\ldots, f_m)$
contains a nonempty Zariski-open set. Moreover, if $(f_1,\ldots,f_m)$ is
a bi-regular sequence, then there are no reductions to zero with the extended $F_5$ criterion.
\end{thm}

\noindent{\em Proof.}
Let $G_m$ be a minimal Gröbner basis of $I_{m-1}:f_m$.  The reductions
to zero $(t,f_m)$ which are not detected by the usual $F_5$ criterion
are exactly those such that $t\in \LM(G_m)$ and $t\notin
\LM(I_{m-1})$. We showed that there exists a nonempty Zariski-open subset $O_1$ of $\mathcal{BL}(n_x,n_y)$ such that if $f_m\in O_1$, then $t\in \LM(I_{m-1}:f_m \cap k[x_0,\ldots, x_{n_x}])$ or
$t\in \LM(I_{m-1}:f_m \cap k[y_0,\ldots, y_{n_y}])$ (Proposition
\ref{propxy}). If we suppose that the conjecture \ref{conjec} is true,
then there exists a nonempty Zariski-open subset $O_2$ of $\mathcal{BL}(n_x,n_y)$ such that if $f_m\in O_2$, $I_{m-1}:f_m \cap k[x_0,\ldots, x_{n_x}]$
(resp. $I_{m-1}:f_m \cap k[y_0,\ldots, y_{n_y}]$) is spanned by the
maximal minors of $\jac_\x(F_{m-1})$ (resp. $\jac_\y(F_{m-1})$). Thus, by Theorem \ref{minorsGB},  there exists a nonempty Zariski-open subset $O_3$ of $\mathcal{BL}(n_x,n_y)$ such that if $f_m\in O_3$,
$\LM(I_{m-1}:f_m \cap k[x_0,\ldots, x_{n_x}])=\mathsf{Monomials}^\x_{m-n_y-2}(n_y+1)\rangle$ (resp. $\LM(I_{m-1}:f_m \cap k[y_0,\ldots, y_{n_y}])=\mathsf{Monomials}^\y_{m-n_x-2}(n_x+1)\rangle$). Suppose that
 $f_m\in O_1\cap O_2\cap  O_3$ (which is a nonempty Zariski-open subset) and that $(t,f_m)$ is a reduction to zero such that $t\notin \LM(I_{m-1})$. Then 

$$t\in\langle \mathsf{Monomials}^\x_{m-n_y-2}(n_y+1)\rangle$$
$$or$$
$$t\in\langle \mathsf{Monomials}^\y_{m-n_x-2}(n_x+1)\rangle.$$
By Lemma \ref{lemGB1}, $t$ is a leading monomial of a linear combination of the maximal minors of $\jac_\x(F_{m-1})$ (or $\jac_\y(F_{m-1})$). Consequently, the reduction to zero $(t,f_m)$ is detected by the extended $F_5$ criterion.
\hfill $\square$

\begin{rem}\label{remarque}
Thanks to the analysis of Algorithm \ref{algoprinc}, we know
exactly which reductions to zero can be avoided during the computation
of a Gröbner basis of a bilinear system. If a bilinear system is
bi-regular, then the Algorithm \ref{algoprinc} finds all reductions to
zero. Indeed, this algorithm detects reductions to zero coming from
linear combinations of maximal minors of the matrices $\jac_\x(F_i)$ and
$\jac_\y(F_i)$. According to Theorem \ref{genericity}, there are no other reductions to zero for
bi-regular systems.
\end{rem}

\section{Hilbert bi-series of bilinear systems}

 An important tool to describe ideals spanned by bilinear
 equations is the so-called \emph{Hilbert series}. In the homogeneous
 case, complexity results for $F_5$ were obtained with this tool (see
 e.g. \cite{bardet2005}). In this section, we provide an explicit form of
 the Hilbert bi-series -- a bihomogeneous analog of the Hilbert
 series -- for ideals spanned by generic bilinear systems. To find
 this bi-series, we use the combinatorics of the syzygy module of
 bi-regular systems. With this tool, we will be able to do a
 complexity analysis of a special version of the $F_5$ which will be
 presented in the next section.

We say that an ideal is \emph{bihomogeneous} if there exists a set of
bihomogeneous generators.  The following notation will be used
throughout this paper: the vector space of bihomogeneous polynomials
of bi-degree $(\alpha,\beta)$ will be denoted by
$R_{\alpha,\beta}$. If $I$ is a bihomogeneous ideal, then
$I_{\alpha,\beta}$ will denote the vector space $I\cap
R_{\alpha,\beta}$.

\begin{defn}[\cite{Waerden,eldin2006}]
Let $I$ be a bihomogeneous ideal of $R$. The Hilbert bi-series is defined by 
$$\HS_I(t_1,t_2)=\sum_{(\alpha,\beta)\in\mathbb{N}^2} \dim(R_{\alpha,\beta}/I_{\alpha,\beta}) t_1^\alpha t_2^\beta.$$
\end{defn}
\begin{rem}
The usual univariate Hilbert series for homogeneous ideals can easily be deduced from the Hilbert bi-series by putting $t_1=t_2$ (see \cite{eldin2006}).
\end{rem}

We can now present the main result of this section: an explicit form of the bi-series for bi-regular bilinear systems.

\begin{thm}\label{theohilb}
Let $f_1,\ldots,f_m\in R$ be a bi-regular bilinear sequence, with $m\leq n_x+n_y$.
Then $$\HS_{I_m}(t_1,t_2)=\frac{N_m(t_1,t_2)}{(1-t_1)^{n_x+1} (1-t_2)^{n_y+1}},$$
where
$$\begin{array}{c}
N_m(t_1,t_2)~~~=~~~(1-t_1 t_2)^m +\\
\sum_{\ell=1}^{m-(n_y+1)}(1-t_1 t_2)^{m-(n_y+1)-\ell} t_1 t_2 (1-t_2)^{n_y+1}\big{[}1-(1-t_1)^\ell \sum_{k=1}^{n_y+1} t_1^{n_y+1-k}{{\ell+n_y-k}\choose {n_y+1-k}}\big{]} +\\
\sum_{\ell=1}^{m-(n_x+1)}(1-t_1 t_2)^{m-(n_x+1)-\ell} t_1 t_2 (1-t_1)^{n_x+1}\big{[}1-(1-t_2)^\ell \sum_{k=1}^{n_x+1} t_2^{n_x+1-k}{{\ell+n_x-k}\choose {n_x+1-k}}\big{]}.
\end{array}$$
\end{thm}
We decompose the proof of this theorem into a sequence of lemmas.\\

If $I$ is an ideal of $R$ and $f$ is a polynomial, we denote by $\bar{f}$ the equivalence class of $f$ in $R/I$ and 
$$\ann_{R/I}(f)=\{v\in R/I : v \bar{f}=0 \},$$
$$\ann_{R/I}(f)_{\alpha,\beta}=\{v\in R/I \text{ of bi-degree }(\alpha,\beta) : v \bar{f}=0\}.$$
If $I$ is a bihomogeneous ideal and $f$ is a bihomogeneous polynomial, we use the following notation:
$$G_{I,f}(t_1,t_2)=\sum_{(\alpha,\beta)\in\mathbb{N}^2}\dim (\ann_{R/I}(f)_{\alpha,\beta}) t_1^\alpha t_2^\beta.$$
\begin{lem}\label{lemmeH}
Let $f_1,\ldots,f_m\in R$ be bihomogeneous polynomials, with $1<m\leq n_x+n_y$.
Let $(d_1,d_2)$ be the bi-degree of $f_m$.
Then 
$$\HS_{I_m}(t_1,t_2)=(1-t_1^{d_1} t_2^{d_2}) \HS_{I_{m-1}}+ t_1^{d_1} t_2^{d_2} G_{I_{m-1},f}(t_1,t_2).$$
\end{lem}

\noindent{\emph{Proof.}}
We have the following exact sequence:
$$0\rightarrow \ann_{R/I_{m-1}}(f)\xrightarrow{\varphi_1} R/I_{m-1}
\xrightarrow{\varphi_2} R/I_{m-1}\xrightarrow{\varphi_3}
R/I_m\rightarrow 0.$$ where $\varphi_1$ and $\varphi_3$ are the
canonical inclusions, and $\varphi_2$ is the multiplication by $f_m$.

From this exact sequence of ideals, we can deduce an exact sequence of
vector spaces:
$$0\rightarrow (\ann_{R/I_{m-1}}(f))_{\alpha,\beta}\xrightarrow{\varphi_1}
\left (\frac{R}{I_{m-1}}\right )_{\alpha,\beta}
\xrightarrow{\varphi_2} \left (\frac{R}{I_{m-1}}\right
)_{\alpha+d_1,\beta+d_2}\xrightarrow{\varphi_3} \left
  (\frac{R}{I_{m}}\right )_{\alpha+d_1,\beta+d_2}\rightarrow 0.$$ Thus
the alternate sum of the dimensions of vector spaces of an exact
sequence is $0$:
$$\begin{array}{l}
  \dim((\ann_{{R}/{I_{m-1}}}(f))_{\alpha,\beta})-\dim\left (\left (\frac{R}{I_{m-1}}\right )_{\alpha,\beta}\right )+\\\dim\left (\left (\frac{R}{I_{m-1}}\right )_{\alpha+d_1,\beta+d_2}\right )-\dim\left (\left (\frac{R}{I_{m}}\right )_{\alpha+d_1,\beta+d_2}\right )=0.
\end{array}$$
By multiplying this relation by $t_1^\alpha t_2^\beta$ and by summing
over $(\alpha,\beta)$, we obtain the claimed recurrence:
$$\HS_{I_m}(t_1,t_2)=(1-t_1^{d_1} t_2^{d_2}) \HS_{I_{m-1}}+
t_1^{d_1} t_2^{d_2} G_{I_{m-1},f}(t_1,t_2).$$
\hfill$\square$
\begin{lem}\label{lemmeGeneratrice}
Let $f_1,\ldots,f_m\in R$ be a bi-regular bilinear sequence, with $m\leq
n_x+n_y$. Then, for all $2\leq i\leq m$,
$$G_{I_{i-1},f_i}(t_1,t_2)=g^{(i-1)}_x(t_1)+g^{(i-1)}_y(t_2),$$ where
$$g^{(i-1)}_x(t)=\begin{cases} 0\text{ if }i\leq
n_y\\ \frac{1}{(1-t)^{n_x+1}}-\sum_{1\leq j\leq
  n_y+1}\frac{{{i-1-j}\choose{n_y+1-j}}
  t^{n_y+1-j}}{(1-t)^{n_x+n_y-i+2}}
\end{cases}.$$
$$g^{(i-1)}_y(t)=\begin{cases}
0\text{ if }i\leq n_x\\
\frac{1}{(1-t)^{n_y+1}}-\sum_{1\leq j\leq n_x+1}\frac{{{i-1-j}\choose{n_x+1-j}} t^{n_x+1-j}}{(1-t)^{n_x+n_y-i+2}}
\end{cases}.$$
\end{lem}
\noindent{\em Proof.}
Saying that $v\in \ann_{R/I_{i-1}}(f_i)$ is equivalent to saying that
the row with signature $(\LM(v),f_i)$ is not detected by the classical $F_5$ criterion. 
According to Theorem \ref{genericity}, if the system is bi-regular, the reductions to zero corresponding to non-trivial syzygies are exactly:

$$\bigcup_{i=n_x+2}^m \{(t, f_i) : t\in\mathsf{Monomials}^\y_{i-n_x-2}(n_x+1)\} \bigcup_{i=n_y+2}^m \{(t, f_i) : t\in\mathsf{Monomials}^\x_{i-n_y-2}(n_y+1)\}.$$
By Proposition \ref{aff}, we know that if $P\in
k[x_0,\ldots,x_{n_x}]\cap (I_{i-1}:f_i)$
(resp. $k[y_0,\ldots,y_{n_y}]\cap (I_{i-1}:f_i)$), then $\forall j,
y_j P\in I_{i-1}$ (resp. $x_j P\in I_{i-1}$). Thus
$G_{I_{i-1},f_i}(t_1,t_2)$ is the generating bi-series of the
monomials of $k[x_0,\ldots,x_{n_x}]$ which are a multiple of a
monomial of degree $n_y+1$ in $x_0,\ldots,x_{i-n_y-2}$ and of the
monomials of $k[y_0,\ldots,y_{n_y}]$ which are a multiple of a
monomial of degree $n_x+1$ in $y_0,\ldots,y_{i-n_x-2}$. Denote by
$g^{(i-1)}_x(t)$ (resp. $g^{(i-1)}_y(t)$) the generating series of the
monomials of $k[x_0,\ldots,x_{n_x}]$ (resp. $k[y_0,\ldots,y_{n_y}]$)
which are a multiple of a monomial of degree $n_y+1$ (resp. $n_x+1$)
in $x_0,\ldots,x_{i-n_y-2}$ (resp. $y_0,\ldots,y_{i-n_x-2}$).  Then we
have
$$G_{I_{i-1},f_i}(t_1,t_2)=g^{(i-1)}_x(t_1)+g^{(i-1)}_y(t_2).$$
Next we use combinatorial techniques to give an explicit form of $g^{(i-1)}_x(t)$ and $g^{(i-1)}_y(t)$.
Let $c(t)$ denote the generating series of the monomials of $k[x_{i-n_y-1},\ldots,x_{n_x}]$: 
$$c(t)=\sum_{j=0}^\infty {{n_x+n_y-i+j+1}\choose j} t^j= \frac{1}{(1-t)^{n_x+n_y-i+2}}.$$ Let $B_j$ denote the number of monomials of $k[x_0,\ldots,x_{i-n_y-2}]$ of degree $j$. Then
$$\frac{1}{(1-t)^{n_x+n_y+2}}=c(t)+B_1 c(t)+\dots+B_{n_y} c(t) +g^{(i-1)}_x(t).$$
Since $B_j={{i-n_y-1+j}\choose j}$, we can conclude:
$$g^{(i-1)}_x(t)=\begin{cases} 0\text{ if }i\leq
n_y\\ \frac{1}{(1-t)^{n_x+1}}-\sum_{1\leq j\leq
  n_y+1}\frac{{{i-1-j}\choose{n_y+1-j}}
  t^{n_y+1-j}}{(1-t)^{n_x+n_y-i+2}}
\end{cases}.$$
\hfill $\square$

\noindent{\em Proof of Theorem \ref{theohilb}.} 
Since the polynomials are bilinear, by Lemma \ref{lemmeH}, we have
$$\HS_{I_i}(t_1,t_2)=(1-t_1 t_2) \HS_{I_{i-1}}+ t_1 t_2 G_{I_{i-1},f_i}(t_1,t_2).$$
Lemma \ref{lemmeGeneratrice} gives the value of $G_{I_{i-1},f_i}(t_1,t_2)$.
To initiate the recurrence, we need
$$\HS_{I_0}(t_1,t_2)=\HS_{\langle 0\rangle}(t_1,t_2)=\frac{1}{(1-t_1)^{n_x+1} (1-t_2)^{n_y+1}}.$$
Then we can obtain the claimed form of the bi-series by solving the recurrence:
$$\HS_{I_i}(t_1,t_2)=\frac{N_i(t_1,t_2)}{(1-t_1)^{n_x+1} (1-t_2)^{n_y+1}}$$
 $$N_i(t_1,t_2)=(1-t_1 t_2)^i + \sum_{j=0}^{m-1}t_1 t_2 (1-t_1 t_2)^{j} G_{I_j,f_{j+1}}(t_1,t_2).$$

\hfill $\square$

\section{Towards complexity results}

\subsection{A multihomogeneous $F_5$ Algorithm}
We now describe how it is possible to use the multihomogeneous
structure of the matrices arising in the Matrix $F_5$ Algorithm to
speed-up the computation of a Gröbner basis. In order to have simple
notations, the description is made in the context of bihomogeneous
systems, but it can be easily transposed in the context of
multihomogeneous systems.

Let $f_1,\ldots,f_m$ be a sequence of bihomogeneous polynomials.  Then
consider the matrices $M_d$ in degree $d$ appearing during the Matrix
$F_5$ Algorithm. One can remark that each row represents a
bihomogeneous polynomial. Let $(d_1,d_2)$ be the bi-degree of one row
of this matrix. Then the only non-zero coefficients on this row are in columns which represent a monomial of bi-degree
$(d_1,d_2)$. Then a possible strategy to use the bihomogeneous
structure is the following:
\begin{itemize}
\item For each couple $(d_1,d_2)$ such that $d_1+d_2=d$, construct the
  matrix $M_{d_1,d_2}$. The rows of this matrix represent the
  polynomials of $M_d$ of bi-degree $(d_1,d_2)$ and the columns represent the monomials of $R_{d_1,d_2}$.
\item Compute the row echelon forms of the matrices
  $M_{d_1,d_2}$. This gives bases of $I_{d_1,d_2}$.
\item The union of the bases gives a basis of $I_d$ since
  $I_d=\bigoplus_{d_1+d_2=d}I_{d_1,d_2}$.
\end{itemize}

This way, instead of computing the row echelon form of a big matrix,
we can decompose the problem and compute independently the row echelon forms of
smaller matrices. This strategy can be extended to multihomogeneous systems. \\

In Table \ref{table:mhomF5}, the execution time and the memory usage
of this multihomogeneous variant of $F_5$ are compared to the
classical homogeneous Matrix $F_5$ Algorithm for computing a
$D$-Gröbner basis for random bihomogeneous systems (for the grevlex
ordering). Both implementations are made in \texttt{Magma2.15-7}. The
experimental results have been obtained with a Xeon processor 2.50GHz
cores and 20 GB of RAM. We are aware that we should compare efficient
implementations of these two algorithms to have a more precise
evaluation of the speed-up we can expect for practical
applications. However, these experiments give a first estimation of that
speed-up. Furthermore, we can also expect to save a lot of memory by decomposing
the Macaulay matrix into smaller matrices. This is crucial for
practical applications, since untractability is often due to the lack of memory.

\begin{center}
\begin{table}
\begin{tabular}{|c|c|c|c|c|c|c|c|c|c|}
\cline{6-9}
\multicolumn{5}{c|}{}&\multicolumn{2}{|c|}{\textbf{Multihomogeneous}}&\multicolumn{2}{|c|}{Homogeneous}&\multicolumn{1}{|c}{}\\
\hline
$n_x$&$n_y$&$m$&bidegree& $D$&time&memory&time&memory&speed-up\\
\hline
\hline
$3$&$4$&$7$&$(1,1)$&$6$&\textbf{16.9s}&\textbf{30MB}&265.7s&280MB&\textbf{16} \\
$3$&$4$&$7$&$(1,1)$&$7$&\textbf{105s}&\textbf{92MB}&2018s&1317MB&\textbf{19} \\ %new2
$4$&$4$&$8$&$(1,1)$&$7$&\textbf{582s}&\textbf{275MB}&13670s&4210MB&\textbf{23} \\
$5$&$4$&$9$&$(1,1)$&$7$&\textbf{3343s}&\textbf{957MB}&66371s&12008MB&\textbf{20} \\
$5$&$5$&$10$&$(1,1)$&$6$&\textbf{645s}&\textbf{435MB}&10735s&4330MB&\textbf{17} \\ %new
\hline
$2$&$2$&$4$&$(1,2)$&$10$&\textbf{11.4s}&\textbf{19MB}&397s&299MB&\textbf{35} \\
$2$&$2$&$4$&$(1,2)$&$8$&\textbf{1.7s}&\textbf{10MB}&16s&52MB&\textbf{9} \\ %new2
$3$&$3$&$6$&$(1,2)$&$8$&\textbf{67s}&\textbf{80MB}&1146s&983MB&\textbf{17} \\%new
$4$&$4$&$8$&$(1,2)$&$8$&\textbf{2222s}&\textbf{1031MB}&40830s&12319MB&\textbf{63} \\%new
\hline
$2$&$2$&$4$&$(2,2)$&$11$&\textbf{29s}&\textbf{27MB}&899s&553MB&\textbf{31} \\
$3$&$3$&$6$&$(2,2)$&$8$&\textbf{27s}&\textbf{47MB}&277s&452MB&\textbf{10} \\%new
$3$&$3$&$6$&$(2,2)$&$9$&\textbf{152s}&\textbf{154MB}&2380s&1939MB&\textbf{16} \\%new2
$3$&$4$&$7$&$(2,2)$&$9$&\textbf{1034s}&\textbf{505MB}&18540s&7658MB&\textbf{18} \\%new2
$4$&$4$&$8$&$(2,2)$&$8$&\textbf{690s}&\textbf{385MB}&7260s&4811MB&\textbf{11} \\%new
$4$&$4$&$8$&$(2,2)$&$9$&\textbf{6355s}&\textbf{2216MB}&---&$>$20000MB&--- \\%new2
\hline
\end{tabular}
\caption{Execution time and memory usage of the multihomogeneous variant of $F_5$}
\label{table:mhomF5}
\end{table}
\end{center}

%%\comm{compléter tab + plus de benchs}
\subsection{A theoretical complexity analysis in the bilinear case}
In this section, we provide a theoretical explanation of the speed-up
observed when using the bihomogeneous structure of bilinear
systems. To estimate the complexity of the Matrix $F_5$ Algorithm, we
consider that the cost is dominated by the cost of the reductions of
the matrices with the highest degree. By using the new criterion
described in Section \ref{sectionalgo}, all the matrices appearing during the
computations have full rank for generic inputs (these ranks are the dimensions of the $k$-vector spaces $I_{d_1,d_2}$).  We consider that the complexity of reducing
a $r\times c$ matrix with Gauss elimination is
$\mathcal{O}(r^2 c)$.
Thus the complexity of computing a $D$-Gröbner basis with the usual Matrix $F_5$ Algorithm and the extended criterion for a bilinear system of $m$ equations over $k[x_0,\ldots,x_{n_x},y_0,\ldots,y_{n_y}]$ is
$$T_{hom}=C_1 \left(\left ({{D+n_x+n_y+1}\choose {D}}-[t^D]\HS(t,t)\right )^2 {{D+n_x+n_y+1}\choose D}\right).$$
When using the multihomogeneous structure, the complexity becomes:
$$T_{multihom}=\displaystyle{C_2\left(\sum_{\tiny\begin{array}{c} d_1+d_2=D\\1\leq d_1,d_2\leq D-1\end{array}} \left (\dim(R_{d_1,d_2}) -[t_1^{d_1} t_2^{d_2}]\HS(t_1,t_2)\right )^2 \dim(R_{d_1,d_2})\right),}$$
where $\dim(R_{d_1,d_2})={{d_1+n_x}\choose {d_1}}{{d_2+n_y}\choose {d_2}}$.
Thus the theoretical speed-up that we expect is:
$$speedup_{th}=C_3 F(n_x,n_y,m,D)$$
where $C_3=\frac{C_1}{C_2}$ is a constant and
$$F(n_x,n_y,m,D)= \left(\frac{\left({{D+n_x+n_y+1}\choose {D}}-[t^D]\HS(t,t)\right)^2 {{D+n_x+n_y+1}\choose D}}{\displaystyle\sum_{\tiny\begin{array}{c} d_1+d_2=D\\1\leq d_1,d_2\leq D-1\end{array}}\left(\dim(R_{d_1,d_2})-[t_1^{d_1} t_2^{d_2}]\HS(t_1,t_2)\right)^2 \dim(R_{d_1,d_2})}\right).$$

Now let us compare this theoretical speed-up with the one observed in practice.\\

\begin{center}
\begin{tabular}{|c|c|c|c|c|c|}
\hline
$n_x$&$n_y$&$m$& $D$&$\begin{array}{c}experimental\\speed-up\end{array}$&$F(n_x,n_y,m,D)$\\
\hline
\hline
$3$&$4$&$7$&$6$&\textbf{16}&\textbf{29}\\
$3$&$4$&$7$&$7$&\textbf{19}&\textbf{34}\\
$4$&$4$&$8$&$7$&\textbf{23}&\textbf{34}\\
$5$&$4$&$9$&$7$&\textbf{20}&\textbf{32}\\
$5$&$5$&$10$&$6$&\textbf{17}&\textbf{27}\\

\hline
\end{tabular}
\end{center}~\\

We can see in this table that, in practice, experimental results match the theoretical complexity: 
$$speedup\approx 0.6 F(n_x,n_y,m,D).$$

\subsection{Structure of generic affine bilinear systems}
In this section, we show that generic \emph{affine} bilinear systems have a
particular structure: they are regular (Definition
\ref{def:regular}). Consequently, the usual $F_5$ criterion removes all
reductions to zero.

\begin{prop}\label{affreg}
Let $S$ be the set of affine bilinear systems over $k[x_1,\ldots,
  x_{n_x},y_1,\ldots, y_{n_y}]$ with $m\leq n_x+n_y$ equations. Then the subset $$\{(f_1,\ldots, f_m)\in S~~:~~(f_1,\ldots, f_m)\text{ is a regular sequence}\}$$ contains a Zariski nonempty open subset of $S$.
\end{prop}

\noindent{\em Proof.}
Let $(f_1,\ldots, f_m)$ be a generic affine bilinear system. Assume
that it is not regular. Then for some $i$, there exists $g\in R$ such
that $g\notin I_{i-1}$ and $g f_i\in I_{i-1}$. Denote by $g^h$ the
bi-homogenization of $g$. Then $g^h \in \langle
f^h_1,\ldots,f^h_{i-1}\rangle : f^h_i$.  $(f^h_1,\ldots,f^h_m)$ is a generic
bilinear system, hence it is bi-regular (Theorem
\ref{genericity}). Thus $g^h\in k[x_0,\ldots, x_{n_x}]$ or $g^h\in
k[y_0,\ldots, y_{n_y}]$. Let us suppose that $g^h\in k[x_0,\ldots,
  x_{n_x}]$ (the proof is similar if $g^h\in k[y_0,\ldots,
  y_{n_y}]$). Therefore $y_{n_y} g^h\in \langle
f^h_1,\ldots,f^h_{i-1}\rangle$ when the system is bi-regular
(Proposition \ref{aff}). By putting $x_{n_x}=1$ and $y_{n_y}=1$, we see that
in this case, $g\in I_{i-1}$, which yields a contradiction. This shows
that generic affine bilinear systems are regular. \begin{flushright} $\square$\end{flushright}

\subsection{Degree of regularity of affine bilinear systems}
In this part, $m$, $n_x$ and $n_y$ are three integers such that
$m=n_x+n_y$. We consider a system of bilinear polynomials
$F=(f_1,\ldots,f_m)\in
k[x_0,\ldots,x_{n_x},y_0,\ldots,y_{n_y}]^m$. $\vartheta$ denotes the deshomogenization morphism:
$$\begin{array}{ccc}
k[x_0,\ldots,x_{n_x},y_0,\ldots,y_{n_y}]&\longrightarrow&k[x_0,\ldots,x_{n_x-1},y_0,\ldots,y_{n_y-1}]\\
f(x_0,\ldots,x_{n_x},y_0,\ldots,y_{n_y})&\longmapsto&f(x_0,\ldots,x_{n_x-1},1,y_0,\ldots,y_{n_y-1},1)\end{array}.$$

Also, $I$ stands for the ideal $\langle f_1,\ldots,f_m\rangle$ and
$\vartheta(I)$ denotes the ideal $\langle
\vartheta(f_1),\ldots,\vartheta(f_m)\rangle$.  In the following, we
suppose without loss of generality that $n_x\leq n_y$.  We also assume
in this part of the paper that the characteristic of $k$ is $0$
(although the results remain true when the
characteristic is large enough).

The goal of this section is to give an upper bound on the so-called
\emph{degree of regularity} of an ideal $I$ generated by a generic
affine bilinear system with $m$ equations and $m$ variables. The
degree of regularity is a crucial indicator of the complexity of
Gröbner bases algorithms: for $0$-dimensional ideals, it is the
lowest integer $d_{reg}$ such that all monomials of degree $d_{reg}$
are in $\LM(I)$ (see \cite{bardet2005}). As a consequence, the degrees of all polynomials occurring
in the $F_5$ algorithm are lower than $d_{reg}+1$. In the following, $\prec$ still denotes the grevlex ordering.

\begin{lem}\label{lem:shapex}
If the system $F$ is generic, then there exists polynomials $g_0,\ldots,g_{n_x-1}\in k[y_0,\ldots,y_{n_y-1}]$ such that
$$\forall j \in \{0,\ldots,n_x-1\}, x_j - g_j(y_0,\ldots, y_{n_y-1})\in \vartheta(I).$$
\end{lem}

\noindent \emph{Proof. }
We consider the $m\times n_x$ matrix $A=\jac_\x(\vartheta(F))$ and the vector
$$B=\begin{pmatrix}\vartheta(f_1)(0,\ldots,0,y_0,\ldots,y_{n_y-1})&\dots&\vartheta(f_m)(0,\ldots,0,y_0,\ldots,y_{n_y-1})\end{pmatrix}.$$
Thus $A\cdot \begin{pmatrix}
x_0\\\vdots\\x_{n_x-1}
\end{pmatrix} + B=\begin{pmatrix}
\vartheta(f_1)\\\vdots\\\vartheta(f_m)
\end{pmatrix}.$

We denote by $\{A^{(i)}\}$ all the $n_x \times n_x$ sub-matrices of $A$.\\
Let $(\alpha_0,\ldots,\alpha_{n_y-1})\in Var(\langle \mathsf{MaxMinors}(\vartheta(\jac_\x(F)))\rangle)$ be an element of the variety.
Let $A_\mathbf{\alpha}$ (resp. $B_\mathbf{\alpha}$) denote the matrix $A$ (resp. $B$) where $y_i$ has been substituted by $\alpha_i$ for all $i$.
Since $\vartheta(I)$ is 0-dimensional, the affine linear system
$$A_\mathbf{\alpha}\cdot \begin{pmatrix}
x_0\\\dots\\x_{n_x-1}
\end{pmatrix}  +B_\mathbf{\alpha}=0 $$
has a unique solution. Therefore, the matrix $A_\mathbf{\alpha}$ is of full rank. Consequently, there exists an invertible $n_y\times n_y$ sub-matrix of $A_\mathbf{\alpha}$.

\smallskip

Since $k$ is infinite, we can suppose without loss of generality that,
if the system is generic, then for all $\alpha$, the matrix
$A_\alpha^{(1)}$ obtained by considering the $n_y$ first columns of
$A_\alpha$ is invertible (if $A_\alpha^{(1)}$ is not invertible, just replace the original bilinear
system by an equivalent system where each new equation is a generic linear
combination of the original equations). Thus $\det(A_\alpha^{(1)})\neq
0$.

\smallskip

According to Lemma \ref{lem:degproj} and \ref{lem:elimJac}, $\langle
\mathsf{MaxMinors}(\vartheta(\jac_\x(F)))\rangle=\langle\vartheta(f_1),\ldots,\vartheta(f_m)\rangle\cap
k[y_0,\ldots,y_{n_y-1}]$. Thus $\det(A^{(1)})$ (i.e. the matrix of the
$n_y$ first columns of $A$) does not vanish on any elements of the
variety of $\vartheta(I)$. Therefore, the Nullstellensatz says that
$\det(A^{(1)})$ is invertible in
$k[y_0,\ldots,y_{n_y-1}]/(\vartheta(I)\cap k[y_0,\ldots,y_{n_y-1}])
$. Let $h$ denotes its inverse. We know from Cramer's rule that there
exists polynomials $g_j\in k[y_0,\ldots,y_{n_y-1}]$ such that
$$x_j \det(A^{(1)})- g_j(y_0,\ldots,y_{n_y-1})\in \vartheta(I).$$
Multiplying this relation by $h$, we obtain:
$$x_j-h g_j(y_0,\ldots,y_{n_y-1}) \in \vartheta(I).\hspace{1cm}\square$$
\begin{thm}\label{thm:degregaff}
If the system $F$ is generic, then the degree of regularity of $\vartheta(I)$ is upper bounded by
$$d_{reg}\leq \min(n_x+1,n_y+1).$$
\end{thm}

\noindent{\em Proof.}
We supposed that $n_x\leq n_y$, so we want to prove that $d_{reg}=n_x+1$.
Let $t=\prod_{j=0}^{n_x-1} x_j^{\alpha_j}\prod_{k=0}^{n_y-1} y_k^{\beta_k}$ be a monomial of degree $n_x+1$.
According to Lemma \ref{lem:shapex},
$$t-\prod_{j=0}^{n_x-1} g_j(y_0,\ldots, y_{n_y-1})^{\alpha_j}\prod_{k=0}^{n_y-1} y_k^{\beta_k} \in \vartheta(I).$$
Now consider the normal form with respect to the ideal $\langle \mathsf{MaxMinors}(\vartheta(\jac_\x(F)))\rangle$. Then
$$t-\NF(\prod_{j=0}^{n_x-1} g_j(y_0,\ldots, y_{n_y-1})^{\alpha_j}\prod_{k=0}^{n_y-1} y_k^{\beta_k}) \in \vartheta(I).$$
Since all monomials of degree $n_x+1$ are in $\LM(\langle \mathsf{MaxMinors}(\vartheta(\jac_\x(F)))\rangle)$ (Lemma \ref{lemGB1}), 
$$\deg(\NF(\prod_{j=0}^{n_x-1} g_j(y_0,\ldots, y_{n_y-1})^{\alpha_j}\prod_{k=0}^{n_y-1} y_k^{\beta_k}))<n_x+1.$$
This implies that
$$\LM(t-\NF(\prod_{j=0}^{n_x-1} g_j(y_0,\ldots, y_{n_y-1})^{\alpha_j}\prod_{k=0}^{n_y-1} y_k^{\beta_k}))=t.$$
Therefore, for each monomial $t$ of degree $n_x+1$, $t\in \LM(\vartheta(I))$. This means that $d_{reg}\leq n_x+1$.
\hfill $\square$

\begin{rem}
This bound on the degree of regularity should be compared with the
degree of regularity of a generic quadratic system with $m$ equations
and $m$ variables. The Macaulay bound (see \cite{lazard83}) says that
the degree of regularity of such systems is $m+1$. Since Gröbner bases
algorithms are exponential in the value, it means that affine bilinear
systems are a lot easier to solve than generic affine quadratic
systems. Moreover, the inequality $d_{reg}\leq \min(n_x+1,n_y+1)$ is
sharp: experimentally, it is an equality for random bilinear systems.
\end{rem}
The following Corollary is a consequence of Theorem \ref{thm:degregaff}.
\begin{cor}\label{coro:complaff}
The arithmetic complexity of computing a Gröbner basis of a generic bilinear system $f_1,\ldots, f_{n_x+n_y}\in k[x_0,\ldots,x_{n_x-1},y_0,\ldots,y_{n_y-1}]$ with the $F_5$ Algorithm is upper bounded by
$$O\left({{n_x+n_y+\min(n_x+1,n_y+1)}\choose{\min(n_x+1,n_y+1)}}^\omega\right),$$
where $2\leq \omega\leq 3$ is the linear algebra constant.
\end{cor}

\begin{proof}
According to \cite{bardet2005}, the complexity of the computation of the Gröbner basis of a $0$-dimensional ideal is upper bounded by 
$$O\left({{n+\mathsf{d_{reg}}}\choose{\mathsf{d_{reg}}}}^\omega\right),$$
  where $n$ is the number of variables and $\mathsf{d_{reg}}$ denotes
  the degree of regularity. In the case of a generic affine bilinear
  system in $k[x_0,\ldots, x_{n_x-1},y_0,\ldots, y_{n_y-1}]$,
  $n=n_x+n_y$ and $\mathsf{d_{reg}}\leq \min(n_x+1,n_y+1)$ (Theorem
  \ref{thm:degregaff}).
\end{proof}

\section{Perspectives and conclusion}
In this paper, we analyzed the structure of ideals generated by
generic bilinear equations. We proposed an explicit description of
their syzygy module. With this analysis, we were able to propose an
extension of the $F_5$ criterion dedicated to bilinear
systems. Furthermore, an explicit formula for the Hilbert bi-series is
deduced from the combinatorics of the syzygy module. With this tool,
we made a complexity analysis of a multihomogeneous variant of the
$F_5$ Algorithm.

We also analyzed the complexity of computing Gröbner bases of affine
bilinear systems. We showed that generic affine bilinear systems are
regular, and we proposed an upper bound for the degree of
regularity of those systems.

Interestingly, properties of the ideals generated by the maximal
minors of the jacobian matrices are especially important. In
particular, a Gröbner basis (for the grevlex ordering) of such an ideal
is a linear combination of the generators. In the affine case, this
ideal permits to eliminate variables.

The next step of this work would be to generalize the results to more
general multihomogeneous systems. For the time being, it is not clear
how the results can be extended. In particular, it would be
interesting to understand the structure of the syzygy module of
general multihomogeneous systems, and to have an explicit formula of
their Hilbert series. Also, having sharp upper bounds on the degree of
regularity of multihomogeneous systems would be important for practical applications.

\nocite{*}
\bibliographystyle{plain}
\bibliography{../Arxiv/bibliobilin}

\begin{thebibliography}{10}

\bibitem{adams}
W.W. Adams and P.~Loustaunau.
\newblock {\em {An introduction to Gr{\"o}bner bases}}.
\newblock American Mathematical Society, 1994.

\bibitem{ars}
G.~Ars.
\newblock {\em Applications des bases de {G}röbner à la cryptographie}.
\newblock PhD thesis, Université de Rennes I, 2005.

\bibitem{AtMc}
M.~Atiyah and I.~MacDonald.
\newblock {\em Introduction to Commutative Algebra}.
\newblock Series in Mathematics. Addison-Wesley, 1969.

\bibitem{bardet}
M.~Bardet.
\newblock {\em Étude des systèmes algébriques surdéterminés. Applications aux
  codes correcteurs et à la cryptographie}.
\newblock PhD thesis, Université Paris 6, 2004.

\bibitem{bardet2005}
M.~Bardet, J.-C. Faug\`ere, B.~Salvy, and B.Y. Yang.
\newblock {Asymptotic behaviour of the degree of regularity of semi-regular
  polynomial systems}.
\newblock In {\em Proceedings of Effective Methods in Algebraic Geometry
  (MEGA)}, 2005.

\bibitem{bardet2004complexity}
M.~Bardet, J.C. Faugere, and B.~Salvy.
\newblock {On the complexity of Gr{\"o}bner basis computation of semi-regular
  overdetermined algebraic equations}.
\newblock In {\em Proceedings of the International Conference on Polynomial
  System Solving}, pages 71--74, 2004.

\bibitem{bernstein93}
David Bernstein and Andrei Zelevinsky.
\newblock Combinatorics of maximal minors.
\newblock {\em Journal of Algebraic Combinatorics}, 2(2):111--121, 1993.

\bibitem{bruns2003grobner}
W.~Bruns and A.~Conca.
\newblock {Gr{\"o}bner bases and determinantal ideals}.
\newblock {\em Arxiv preprint math/0302058}, 2003.

\bibitem{buchberger}
B.~Buchberger.
\newblock {An algorithm for finding the basis elements of the residue class
  ring of a zero dimensional polynomial ideal}.
\newblock {\em Journal of Symbolic Computation}, 41(3-4):475--511, 2006.

\bibitem{cox}
D.~Cox, J.~Little, and D.~O'Shea.
\newblock {\em {Ideals, Varieties, and Algorithms}}.
\newblock Springer, 2007.

\bibitem{multires}
A.~Dickenstein and I.Z. Emiris.
\newblock Multihomogeneous resultant formulae by means of complexes.
\newblock {\em Journal of Symbolic Computation}, 36(3-4):317--342, 2003.

\bibitem{eisenbud}
D.~Eisenbud.
\newblock {\em Commutative algebra with a view toward algebraic geometry},
  volume 150 of {\em Graduate Texts in Mathematics}.
\newblock Springer-Verlag, 1995.

\bibitem{EmirisM09}
Ioannis~Z. Emiris and Angelos Mantzaflaris.
\newblock Multihomogeneous resultant formulae for systems with scaled support.
\newblock In {\em Proceedings of the 2009 International Symposium on Symbolic
  and Algebraic Computation}, pages 143--150. ACM, 2009.

\bibitem{fglm}
J.-C. Faug\`ere.
\newblock {\em {R{\'e}solution des systemes d'{\'e}quations alg{\'e}briques.}}
\newblock PhD thesis, Universit{\'e} Paris 6, 1994.

\bibitem{faugere139nea}
J.-C. Faug{\`e}re.
\newblock {A new efficient algorithm for computing Gröbner bases (F4)}.
\newblock {\em Journal of Pure and Applied Algebra}, 139:61--88, 1999.

\bibitem{faugere2002nea}
J.-C. Faug{\`e}re.
\newblock {A new efficient algorithm for computing Gr{\"o}bner bases without
  reduction to zero (F5)}.
\newblock In {\em Proceedings of the 2002 International Symposium on Symbolic
  and Algebraic Computation (ISSAC)}, pages 75--83. ACM New York, NY, USA,
  2002.

\bibitem{faugere2008}
J.-C. Faug{\`e}re, F.~Levy-Dit-Vehel, and L.~Perret.
\newblock {Cryptanalysis of MinRank}.
\newblock In {\em Proceedings of the 28th Annual conference on Cryptology:
  Advances in Cryptology}, pages 280--296. Springer, 2008.

\bibitem{issacmr}
J.-C. Faug\`ere, M.~{Safey El Din}, and P.-J. Spaenlehauer.
\newblock {Computing Loci of Rank Defects of Linear Matrices using Gr\"obner
  Bases and Applications to Cryptology}.
\newblock Submitted to ISSAC 2010, 2010.

\bibitem{faurah}
Jean-Charles Faug{\`e}re and S.~Rahmany.
\newblock Solving systems of polynomial equations with symmetries using
  sagbi-gr{\"o}bner bases.
\newblock In {\em Proceedings of the 2009 International Symposium on Symbolic
  and Algebraic Computation}, pages 151--158. ACM, 2009.

\bibitem{froberg}
R.~Fr\"oberg.
\newblock {\em {An introduction to Gr\"obner bases}}.
\newblock John Wiley \& Sons, 1997.

\bibitem{gabidulin1985tcm}
E.M. Gabidulin.
\newblock {Theory of codes with maximum rank distance}.
\newblock {\em Problemy Peredachi Informatsii}, 21(1):3--16, 1985.

\bibitem{hartshorne1977ag}
R.~Hartshorne.
\newblock {\em {Algebraic geometry}}.
\newblock Springer, 1977.

\bibitem{jeronimo2007computing}
G.~Jeronimo and J.~Sabia.
\newblock {Computing multihomogeneous resultants using straight-line programs}.
\newblock {\em Journal of Symbolic Computation}, 42(1-2):218--235, 2007.

\bibitem{kreuzer2002btc}
M.~Kreuzer and L.~Robbiano.
\newblock Basic tools for computing in multigraded rings.
\newblock In J.~Herzog and V.~Vuletescu, editors, {\em Commutative Algebra,
  Singularities and Computer Algebra}, pages 197--216. Kluwer Academic
  Publishers, 2003.

\bibitem{kreuzer2002basic}
M.~Kreuzer, L.~Robbiano, J.~Herzog, and V.~Vulutescu.
\newblock {Basic tools for computing in multigraded rings}.
\newblock In {\em Commutative Algebra, Singularities and Computer Algebra,
  Proc. Conf. Sinaia}, pages 197--216, 2002.

\bibitem{lazard83}
D.~Lazard.
\newblock Gr{\"o}bner bases, gaussian elimination and resolution of systems of
  algebraic equations.
\newblock In {\em EUROCAL}, pages 146--156, 1983.

\bibitem{bezout}
T.~Li, Z.~Lin, and F.~Bai.
\newblock {Heuristic methods for computing the minimal multi-homogeneous
  B\'ezout number}.
\newblock {\em Applied Mathematics and Computation}, 146(1):237--256, 2003.

\bibitem{matsumura}
H.~Matsumura.
\newblock {\em {Commutative ring theory}}.
\newblock Cambridge Univ Pr, 1989.

\bibitem{mayr1982cwp}
E.W. Mayr and A.R. Meyer.
\newblock {The complexity of the word problems for commutative semigroups and
  polynomial ideals}.
\newblock {\em Adv. Math}, 46(3):305--329, 1982.

\bibitem{mccoy1933resultant}
N.H. McCoy.
\newblock {On the resultant of a system of forms homogeneous in each of several
  sets of variables}.
\newblock {\em Transactions of the American Mathematical Society}, pages
  215--233, 1933.

\bibitem{morgan}
A.~Morgan and A.~Sommese.
\newblock {A homotopy for solving general polynomial systems that respects
  m-homogeneous structures.}
\newblock {\em Appl. Math. Comput.}, 24(2):101--113, 1987.

\bibitem{ourivski2002ntd}
A.V. Ourivski and T.~Johansson.
\newblock {New technique for decoding codes in the rank metric and its
  cryptography applications}.
\newblock {\em Problems of Information Transmission}, 38(3):237--246, 2002.

\bibitem{remond}
G.~R\'emond.
\newblock Elimination multihomog\`ene.
\newblock {\em Introduction to Algebraic Independence Theory. Lect. Notes
  Math}, 1752:53--81, 2001.

\bibitem{remond1752geometrie}
G.~R{\'e}mond.
\newblock {G{\'e}om{\'e}trie diophantienne multiprojective, chapitre 7 de
  Introduction to algebraic independence theory}.
\newblock {\em Lecture Notes in Math}, pages 95--131, 2001.

\bibitem{el2003polar}
M.~{Safey El Din} and E.~Schost.
\newblock {Polar varieties and computation of one point in each connected
  component of a smooth real algebraic set}.
\newblock In {\em Proceedings of the 2003 International Symposium on Symbolic
  and Algebraic Computation}, pages 224--231. ACM New York, NY, USA, 2003.

\bibitem{eldin2006}
M.~{Safey El Din} and P.~Tr\'ebuchet.
\newblock {Strong bi-homogeneous B{\'e}zout theorem and its use in effective
  real algebraic geometry}.
\newblock {\em Arxiv preprint cs/0610051}, 2006.

\bibitem{Shafarevich77}
I.~Shafarevich.
\newblock {\em Basic Algebraic Geometry 1}.
\newblock Springer Verlag, 1977.

\bibitem{sturmfels1993maximal}
B.~Sturmfels and A.~Zelevinsky.
\newblock {Maximal minors and their leading terms}.
\newblock {\em Adv. Math}, 98(1):65--112, 1993.

\bibitem{hilbertdriven}
C.~Traverso.
\newblock {Hilbert functions and the Buchberger algorithm}.
\newblock {\em Journal of Symbolic Computation}, 22(4):355--376, 1996.

\bibitem{Waerden}
Bartel~Leendert {Van der Waerden}.
\newblock {On Hilbert's Function, Series of Composition of Ideals and a
  generalization of the Theorem of Bezout}.
\newblock In {\em Proceedings Roy. Acad. Amsterdam}, volume~31, pages 749--770,
  1929.

\end{thebibliography}
\appendix
\section{Bihomogeneous ideals}
In this part, we use notations similar to those used in Section \ref{sec:gener}:
\begin{itemize}
\item ${\cal BH}(n_x, n_y)$ the $k$-vector space of bilinear
  polynomials in $k[x_0, \ldots, x_{n_x}, y_0, \ldots, y_{n_y}]$;
\item $X\subset k[x_0, \ldots, x_{n_x}, y_0, \ldots, y_{n_y}]$
  (resp. $Y$) is the ideal $\langle x_0, \ldots, x_{n_x}\rangle$
  (resp. $\langle y_0, \ldots, y_{n_y}\rangle$);
\item An ideal is called \emph{bihomogeneous} if there exists a set of bihomogeneous generators. In particular, ideals spanned by bilinear polynomials are bihomogeneous.
\item $J_i$ denotes the saturated ideal $I_i : (X \cap Y)^\infty$;
\item Given a polynomial sequence $(f_1,\ldots, f_m)$, we denote by
  $Syz_{triv}$ the module of trivial syzygies, i.e. the set of all
  syzygies $(s_1,\ldots, s_m)$ such that $\forall 1\leq i\leq m$,
  $s_i\in \langle f_1,\ldots, f_{i-1},f_{i+1},\ldots, f_m\rangle$;
\item A primary ideal $P\subset R$ is called \emph{admissible} if
  $\langle x_0,\ldots, x_{n_x}\rangle \not\subset\sqrt{P}$ and
  $\langle y_0,\ldots, y_{n_y}\rangle \not\subset\sqrt{P}$;
\item Let $E$ be a $k$-vector space such that $\dim(E)<\infty$. We say that a property $\mathcal P$ is \emph{generic} if it is satisfied on a nonempty open subset of $E$ (for the Zariski topology), i.e. $\exists h\in k[\mathfrak a_1,\ldots,\mathfrak a_{\dim(E)}], h\neq 0$, such that 
$$\mathcal P\text{ does not hold on }(a_1,\ldots,a_{\dim(E)})\Rightarrow h(a_1,\ldots,a_{\dim(E)})=0.$$
\end{itemize}

\begin{prop}[\cite{eldin2006}]\label{prop:bihom}
Let $I$ be an ideal of $R$. The two following assertions are equivalent:
\begin{itemize}
\item $I$ is \emph{bihomogeneous}.
\item For all $h\in I$, every bihomogeneous component of $h$ is in $I$.
\end{itemize}
\end{prop}

\begin{lem}[\cite{eldin2006}]\label{altDef}
Let $f_1,\ldots,f_m\in
R$ be polynomials, and $I_m=\cap P_l$ be a minimal primary decomposition of $I_m$ and let
$Adm$ be the set of the admissible ideals of the decomposition.  Then
$J_m=\cap_{P\in Adm}P$.
\end{lem}

\begin{prop}\label{prop:divsat}
  let $f_1,\ldots,f_m\in R$ be polynomials with $m\leq n_x+n_y$, and
  $Ass(I_{i-1})$ be the set of prime ideals associated to
  $I_{i-1}$. The following assertions are equivalent:
\begin{enumerate}
\item $\forall 2\leq i\leq m, f_i$ is not a divisor of $0$ in $R/J_{i-1}$.
\item $\forall 2\leq i\leq m, (f_i\in P, P\in Ass(I_{i-1}))\Rightarrow P\text{ is non-admissible}$. 
\end{enumerate}
\end{prop}
\noindent{\em Proof.}
It is a straightforward consequence of Lemma \ref{altDef}.
\hfill $\square$

\begin{rem}
All results in this section can be generalized to
multihomogeneous systems. Since we focus on bilinear systems in this
paper, we describe them in this more restrictive context.
\end{rem}

\begin{lem}\label{lem:genadm}
Let $P$ be an admissible prime ideal of $R$. The set of bilinear polynomials $f\in R$ such that $f\notin P$ contains a Zariski nonempty open set.
\end{lem}

\noindent{\em Proof.}
Let $f$ be the generic bilinear polynomial 
$$f=\sum_{j,k} \mathfrak a_{j,k} x_j y_k$$
in $k(\{\mathfrak a_{j,k}\}_{0\leq j\leq n_x, 0\leq k\leq
  n_y})[x_0,\ldots,x_{n_x},y_0,\ldots,y_{n_y}]$. Since $P$ is
admissible, there exists $x_{j_0} y_{k_0}$ such that $x_{j_0} y_{k_0}
\notin P$ (this shows the non-emptiness). Let $\prec$ be an admissible
order. Then consider the normal form for this order
$$\NF_P(f)=\sum_{t\text{
    monomial}}h_{t}(\mathfrak a_{0,0}\ldots, \mathfrak a_{n_x,n_y}) t.$$
where the $h_t$'s are polynomials.
Thus, if a polynomial $\tilde{f}\in R$ is in $P$, then its coefficients are in the variety of the polynomial system $\forall t, h_{t}(\mathfrak a_{0,0},\ldots, \mathfrak a_{n_x,n_y})=0$.
\hfill $\square$

\begin{thm}\label{generbireg}
Let $m,n_x,n_y\in \mathbb{N}$ such that $m\leq n_x+n_y$. Then the set of bilinear systems
$f_1,\ldots, f_m$ such that $\forall 2\leq i\leq m, f_i$ is not a divisor of $0$ in $R/J_{i-1}$ contains a Zariski nonempty open subset.
\end{thm}

\noindent{\em Proof.}
We prove the Theorem by recurrence on $m$.
Suppose $\forall 2\leq i\leq m-1$, $f_i$ is not a divisor of $0$ in $R/J_{i-1}$. We prove that the set of bilinear polynomials $f$ such that $f$ is not a divisor of  $0$ in $R/J_{m-1}$ contains a nonempty Zariski open subset. According to Lemma \ref{lem:genadm}, for each admissible prime
ideal $P\in Ass(I_{m-1})$, the set $\mathcal{O}_P=\{f\notin P\}$
contains a nonempty Zariski open subset. Thus $\bigcap_P
\mathcal{O}_P$ contains a nonempty Zariski subset (since the
intersection of a finite number of nonempty Zariski open subsets is a nonempty Zariski open subset). Therefore, the set of bilinear polynomials $f$ which are not divisor of $0$ in $R/J_{m-1}$ (this set is exactly $\bigcap_P
\mathcal{O}_P$) contains a Zariski nonempty open subset.\hfill $\square$

\begin{prop}\label{equidim}
Let $m\leq n_x+n_y$ and $f_1,\ldots, f_m$ be bilinear polynomials such that $\forall 2\leq i\leq m$, $f_i$ is not a divisor of $0$ in $R/J_{i-1}$. Then $\forall 1\leq i\leq m$, the ideal $J_i$ is equidimensional and its co-dimension is $i$.
\end{prop}

\noindent{\em Proof.}
We prove the Proposition by recurrence on $m$.
\begin{itemize}
\item $J_1=I_1$ is equidimensional and $\mathsf{codim}(I_1)=1$;
\item Suppose that $J_{i-1}$ is equidimensional of co-dimension
  $i-1$. Then $J_i=(J_{i-1}+f_i):(X\cap Y)^\infty$. $f_i$ is not
  divisor of $0$ in $J_{i-1}$ (Theorem \ref{generbireg}), thus $J_{i-1}+f_i$ is equidimensional
  of co-dimension $i$. Next, the saturation does not change the
  dimension of any primary component of a minimal primary
  decomposition of $J_{i-1}+f_i$ (the saturation only removes some
  components). Therefore, $J_i$ is equidimensional and its co-dimension is $i$.
\end{itemize}
\hfill $\square$

\section{Ideals generated by generic affine bilinear systems}
Let $k$ be a field of characteristic $0$, $m=n_x+n_y$, and $\mathfrak a$
be the set
$$\mathfrak a=\{\mathfrak a^{(i)}_{j,k} : 1\leq i\leq m, 0\leq j\leq n_x, 0\leq k\leq n_y\}.$$
We consider generic polynomials $f_1,\ldots, f_m$ in $k(\mathfrak
a)[x_0,\ldots, x_{n_x},y_0,\ldots, y_{n_y}]$:
$$f_i=\sum \mathfrak a^{(i)}_{j,k} x_j y_k$$
and we denote by $I\subset k({\frak a})[x_0, \ldots, x_{n_x}, y_0,
\ldots, y_{n_y}]$ the ideal they generate. In the sequel, $\vartheta$
denotes the deshomogeneization morphism:
$$\begin{array}{ccc}
k[x_0,\ldots,x_{n_x},y_0,\ldots,y_{n_y}]&\longrightarrow&k[x_0,\ldots,x_{n_x-1},y_0,\ldots,y_{n_y-1}]\\
f(x_0,\ldots,x_{n_x},y_0,\ldots,y_{n_y})&\longmapsto&f(x_0,\ldots,x_{n_x-1},1,y_0,\ldots,y_{n_y-1},1)\end{array}.$$
For $\mathbf a\in k^{m (n_x+n_y+2)}$, $\varphi_{\mathbf{a}}$ stands for the specialization:
$$\begin{array}{cccc}
\varphi_{\mathbf{a}}:&k(\mathfrak a)[x_0,\ldots, x_{n_x},y_0,\ldots, y_{n_y}]&\rightarrow&k[x_0,\ldots, x_{n_x},y_0,\ldots, y_{n_y}]\\
&f(\mathfrak a)(x_0,\ldots,  x_{n_x},y_0,\ldots, y_{n_y})&\mapsto&f(\mathbf a)(x_0,\ldots,  x_{n_x},y_0,\ldots, y_{n_y})
\end{array}$$

Also $Var(\varphi_{\mathbf{a}}(I))\subset \mathbb P^{n_x}\times \mathbb P^{n_y}$ (resp. $Var(\vartheta\circ\varphi_{\mathbf{a}}(I))\subset \bar k^{n_x+n_y}$) denotes the variety of $\varphi_{\mathbf{a}}(I)$ (resp. $\vartheta\circ\varphi_{\mathbf{a}}(I)$).

\begin{lem}\label{lem:bezoutaffine}
There exists a nonempty Zariski open set $O_1$ such that if $\mathbf
a\in O_1$, then for all
$(\alpha_0,\dots,\alpha_{n_x},\penalty-1000\beta_0,\ldots,\beta_{n_y})\in
Var(\varphi_{\mathbf{a}}(I))$, $\alpha_{n_x}\neq 0$ and
$\beta_{n_y}\neq 0$. This implies that the application
$$\begin{array}{ccc}
Var(\vartheta\circ\varphi_{\mathbf{a}}(I))&\longrightarrow&Var(\varphi_{\mathbf{a}}(I))\\
(\alpha_0,\dots,\alpha_{n_x-1},\beta_0,\ldots,\beta_{n_y-1})&\longmapsto&(\alpha_0,\dots,\alpha_{n_x-1},1,\beta_0,\ldots,\beta_{n_y-1},1)
\end{array}$$
is a bijection.
\end{lem}

\noindent{\em Proof.}
See \cite[page 751]{Waerden}.
\hfill $\square$

\begin{lem}\label{lem:radical}
There exists a nonempty Zariski open set $O_2$, such that if $\mathbf
a\in O_2$, then the ideal $\vartheta\circ\varphi_{\mathbf{a}}(I)$ is
radical.
\end{lem}

\noindent{\em Proof.} Denote by $F$ the polynomial family $(f_1,
\ldots, f_m)$. Let $J\subset k[{\frak a}]$ be the ideal $\left
  (I+\langle \det(\jac(F))\rangle\right )\cap k[{\frak a}]$ and
$\mathscr{J}$ be its associated algebraic variety. By the Jacobian
Criterion (see e.g. \cite[Theorem 16.19]{eisenbud}), if $\mathbf{a}$
does not belong to $\mathscr{J}$, then
$\vartheta\circ\varphi_{\mathbf{a}}(I)$ is radical. Thus, it is
sufficient to prove that $ k^{m(n_x+n_y+2)}\setminus \mathscr{J}$ is
non-empty.

To do that, we prove that for all $\mathbf{a}\in k^{m(n_x+n_y+2)}$,
there exists $(\varepsilon_1, \ldots, \varepsilon_m)$ such that the
ideal $\langle \vartheta\circ\varphi_{\mathbf{a}}(f_1)+\varepsilon_1,
\ldots, \vartheta\circ\varphi_{\mathbf{a}}(f_m)+\varepsilon_m\rangle$
is radical. Denote by $g_i=\vartheta\circ\varphi_{\mathbf{a}}(f_i)$
for $1\leq i \leq m$ and consider the mapping $\Psi$
$$x\in k^m\rightarrow (g_1(x), \ldots, g_m(x))\in k^m.$$
Suppose first that $\Psi(k^m)$ is not dense in $k^m$. Since
$\Psi(k^m)$ is a constructible set, it is contained in a
Zariski-closed subset of $k^m$ and there exists $(\varepsilon_1,
\ldots, \varepsilon_m)$ such that the algebraic variety defined by $
g_1-\varepsilon_1= \cdots=g_m-\varepsilon_m=0$ is empty. Since there
exists $\mathbf{a'}$ such that
$g_i-\varepsilon_i=\vartheta\circ\varphi_{\mathbf{a'}}(f_i)$, we
conclude that $\vartheta\circ\varphi_{\mathbf{a'}}(I)=\langle
1\rangle$. This implies that $\mathbf{a}'\notin\mathscr{J}$.

Suppose now that $\Psi(k^m)$ is dense in $k^m$. By Sard's theorem
\cite[Chap. 6, Theorem 2]{Shafarevich77}, there exists
$(\varepsilon_1, \ldots, \varepsilon_m)\in k^m$ which does not lie in
the set of critical values of $\Psi$.  This implies that at any point
of the algebraic variety defined by
$g_1-\varepsilon_1=\cdots=g_m-\varepsilon_m=0$,
$\vartheta\circ\varphi_{\mathbf{a}}(\det(\jac(F)))$ does not vanish.
Remark now that there exists $\mathbf{a}'$ such that
$g_i-\varepsilon_i=\vartheta\circ\varphi_{\mathbf{a'}}(f_i)$. We
conclude that $\mathbf{a}'\in k^{m(n_x+n_y+2)}\setminus \mathscr{J}$,
which ends the proof.  \hfill $\square$

\begin{lem}\label{lem:elimJac}
There exists a nonempty Zariski open set $O_3$, such that if $\mathbf a\in O_3$,
$$\sqrt{\langle
  \mathsf{MaxMinors}(\vartheta\circ\varphi_{\mathbf{a}}(\jac_\y(F)))\rangle}=\langle\vartheta\circ\varphi_{\mathbf{a}}(f_1),\ldots,\vartheta\circ\varphi_{\mathbf{a}}(f_m)\rangle\cap
k[x_0,\ldots,x_{n_x-1}].$$
\end{lem}

\noindent{\em Proof.}
Let $\mathbf a$ be an element in $O_2$ (as defined in Lemma \ref{lem:radical}). Thus
$\vartheta\circ\varphi_{\mathbf{a}}(I)$ is radical. Now let
$(v_0,\ldots,v_{n_x-1},w_0,\ldots,w_{n_y-1})\in
Var(\vartheta\circ\varphi_{\mathbf{a}}(I))$ be an element of the variety. Then
$$\left(\vartheta\circ\varphi_{\mathbf{a}}(\jac_\y(F))_{x_i=v_i}\right)\cdot\begin{pmatrix}
w_0\\\vdots\\w_{n_y-1}\\1
\end{pmatrix}=\begin{pmatrix}0\\\vdots\\0\end{pmatrix}.$$
This implies that
$\textsf{rank}(\vartheta\circ\varphi_{\mathbf{a}}(\jac_\y(F))_{x_i=v_i})<n_y+1$, and
therefore $$(v_0,\ldots,v_{n_x-1})\in Var(\langle
\mathsf{MaxMinors}(\vartheta\circ\varphi_{\mathbf{a}}(\jac_\y(F)))\rangle).$$

Conversely, let $(v_0,\ldots,v_{n_x-1})\in Var(\langle
\mathsf{MaxMinors}(\vartheta\circ\varphi_{\mathbf{a}}(\jac_\y(F)))\rangle)$. Thus
there exists a non trivial vector $(w_0,\ldots,w_{n_y})$ in the right
kernel
$\Ker(\vartheta\circ\varphi_{\mathbf{a}}(\jac_\y(F))_{x_i=v_i})$. This
means that $(v_0,\ldots,v_{n_x-1},1,\penalty-1000 w_0,\ldots,w_{n_y})$ is in the
variety of $\varphi_{\mathbf{a}}(I)$:
$$(v_0,\ldots,v_{n_x-1},1,w_0,\ldots,w_{n_y})\in Var(
\varphi_{\mathbf{a}}\left(\jac_\y(F)\right)\cdot\begin{pmatrix}
y_0\\\dots\\y_{n_y}
\end{pmatrix})$$
From Lemma \ref{lem:bezoutaffine}, $w_{n_y}\neq
0$ if the system is generic.  Hence
$$(v_0,\ldots,v_{n_x-1},\frac{w_0}{w_{n_y}},\ldots,\frac{w_{n_y-1}}{w_{n_y}})\in
Var(\vartheta\circ\varphi_{\mathbf{a}}(I)).$$

Finally, we have $$Var(\langle
\mathsf{MaxMinors}(\vartheta\circ\varphi_{\mathbf{a}}(\jac_\y(F)))\rangle)=Var(\langle\vartheta\circ\varphi_{\mathbf{a}}(f_1),\ldots,\vartheta\circ\varphi_{\mathbf{a}}(f_m)\rangle\cap
k[x_0,\ldots,x_{n_x-1}])$$ and $\vartheta\circ\varphi_{\mathbf{a}}(I)$
is radical (Lemma \ref{lem:radical}). The Nullstellensatz concludes the proof.  \hfill $\square$

\begin{cor}\label{coro:degaff}
There exists a nonempty Zariski open set $O_4$, such that if $\mathbf a\in O_4$,
$$\mathsf{card}(Var(\vartheta\circ\varphi_{\mathbf{a}}(I)))=\deg(\vartheta\circ\varphi_{\mathbf{a}}(I))={{n_x+n_y}\choose{n_x}}$$
\end{cor}

\noindent{\em Proof.}
According to Lemma \ref{lem:radical} and Lemma \ref{lem:bezoutaffine}, if $\mathbf a\in O_1\cap O_2$, then
$\deg(\vartheta\circ\varphi_{\mathbf{a}}(I))=\mathsf{card}(Var(\vartheta\circ\varphi_{\mathbf{a}}(I))=\mathsf{card}(Var(\varphi_{\mathbf{a}}(I)))$. This value is the so-called
multihomogeneous B\'ezout number of $\varphi_{\mathbf{a}}(I)$, i.e. the coefficient of
$z_1^{n_x} z_2^{n_y}$ in $(z_1+z_2)^{n_x+n_y}$ (see e.g. \cite{morgan}), namely
${{n_x+n_y}\choose{n_x}}$.
\hfill $\square$

\begin{rem}
Actually, by studying ideals spanned by maximal minors of matrices whose entries are linear form, it can be shown that, for a generic affine bilinear system, $\langle
  \mathsf{MaxMinors}(\vartheta\circ\varphi_{\mathbf{a}}(\jac_\y(F)))\rangle$ is radical (see Lemma \ref{lem:degproj}). Hence Lemma \ref{lem:elimJac} shows that, for generic affine bilinear systems,
$$\langle
  \mathsf{MaxMinors}(\vartheta\circ\varphi_{\mathbf{a}}(\jac_\y(F)))\rangle=\langle\vartheta\circ\varphi_{\mathbf{a}}(f_1),\ldots,\vartheta\circ\varphi_{\mathbf{a}}(f_m)\rangle\cap
k[x_0,\ldots,x_{n_x-1}],$$
$$\langle
  \mathsf{MaxMinors}(\vartheta\circ\varphi_{\mathbf{a}}(\jac_\x(F)))\rangle=\langle\vartheta\circ\varphi_{\mathbf{a}}(f_1),\ldots,\vartheta\circ\varphi_{\mathbf{a}}(f_m)\rangle\cap
  k[y_0,\ldots,y_{n_y-1}].$$
\end{rem}

\end{document}